\documentclass[11pt]{article}
\usepackage[T1]{fontenc}
\usepackage{lmodern,amsmath,amssymb,amsthm,mathtools,booktabs,graphicx}
\usepackage[margin=1in]{geometry}
\usepackage[colorlinks=true,linkcolor=blue,citecolor=blue,urlcolor=blue]{hyperref}
\usepackage{microtype}
\newtheorem{theorem}{Theorem}[section]
\newtheorem{lemma}[theorem]{Lemma}
\newtheorem{corollary}[theorem]{Corollary}

\newcommand{\eps}{\varepsilon}
\newcommand{\Prob}{\mathbb P}
\newcommand{\E}{\mathbb E}
\newcommand{\R}{\mathbb R}
\newcommand{\F}{\mathcal F}
\newcommand{\ind}{\mathbf1}
\newcommand{\osc}{\operatorname{osc}}

\allowdisplaybreaks
\title{Restart Degeneration of the Swing Filter and a Simple Repair}
\author{Yue Chen}
\date{September 12, 2026\\\small Working draft}
\begin{document}
\maketitle

\begin{abstract}
The Swing filter extends a feasible line through its previous recording
and restarts from a fitted value at the last accepted sample. We show
that this inherited endpoint causes a singular loss of compression on
rough input. For uniformly sampled Brownian motion at fixed tolerance,
every fixed number of recording times converges to the first continuous
feasibility boundary, with alternating boundary errors. The longest
post-first segment tends to zero in probability. The segment-count ratio
against a continuous approximation through true sampled knots diverges
almost surely whenever that boundary precedes the observation horizon.

For unit Gaussian increments, constrained least squares and no maximum
lag, the first mean segment length has quadratic order in the tolerance
$E$, the second has linear order, and every fixed mean from the third
onward has order $E\log E$. We explain the logarithm through the
conditional survival tail and prove a geometric bound on the subsequent
approach to the stationary coefficient in the large-tolerance limit.
The stationary mean satisfies
$C_E\sim\gamma_{\rm Swing}E\log E$. An explicit one-dimensional
kernel specifies the single coefficient
$\gamma_{\rm Swing}\approx1.8120703$ and admits a $3/4$ Wasserstein
contraction and convergent evaluation brackets. The decimal is not a
certified error interval.

A half-budget slope corridor, true endpoint chords and one-sample
bridges repair the degeneration while preserving the original error
bound, continuous output and constant working memory. Including bridge
costs, the long-stream mean span per output segment is asymptotic to
$21\zeta(3)E^2/(8\pi^2)$. Its explicit coefficient follows from the
known Brownian anchored lifetime. Thus a change to the recording rule
restores the quadratic scale.
\end{abstract}

\section{Introduction}

Online piecewise-linear approximation with a prescribed maximum error is
used to compress numerical streams. The Swing filter of Elmeleegy
et al.~\cite{swing} maintains a slope interval through the preceding
recording. When a sample makes that interval empty, the filter chooses
by constrained least squares a feasible line for the old prefix, records
its value at the preceding sample time, and carries the rejected sample
into the next interval. This recording need not equal the observed value.

That distinction determines the behavior on rough input. At the first
continuous feasibility boundary, the fitted endpoint reaches the error
boundary. Brownian motion started with this inherited error admits no
feasible line for a positive future interval. A discrete filter continues
to advance, but the resulting segments have a singular dense-sampling
limit. We analyze this specific recording and restart rule and give a
simple repair with the same error tolerance and continuous output.

The first segment makes the change particularly visible. For unit Gaussian
increments and large tolerance $E$, its mean is asymptotic to
$21\zeta(3)E^2/\pi^2$. The second segment already loses a power of $E$,
but its mean is $o(E\log E)$, smaller even than the stationary mean.
Thus the first restart and the approach to equilibrium are separate
phenomena. We explain the exceptional second segment before averaging
over the states encountered far along the stream. We also prove that
the logarithmic mean order is present by the third segment and quantify
its subsequent approach to the stationary coefficient.

Our main conclusions are as follows.
\begin{enumerate}
\item For fixed-tolerance Brownian sampling, the recording times collapse
onto the first feasibility boundary and the endpoint errors alternate
between the two error boundaries. The longest post-first segment tends
to zero in probability. The number of segments nevertheless diverges
almost surely after the first boundary, whereas an explicit continuous
approximation through true sampled knots has uniformly bounded count.
These conclusions hold for every causal feasible slope choice.
\item For native constrained least squares with unit Gaussian increments,
$\E L_{2,E}=\Theta(E)$ and $\E L_{k,E}=\Theta_k(E\log E)$ for
every fixed $k\ge3$. The normalized later means approach one stationary
coefficient with a geometric large-tolerance bound. At each fixed $E$
the restart chain has a unique stationary law, whose mean satisfies
$C_E\sim\gamma_{\rm Swing}E\log E$. We derive an explicit limiting
kernel and a $3/4$ Wasserstein contraction, which also gives convergent
brackets for evaluating $\gamma_{\rm Swing}$. These bounds do not
assert a mixing time for the finite-$E$ stream.
\item A half-budget corridor, followed by a true endpoint chord and a
one-sample bridge, preserves the final error bound with $O(1)$ processing
per observation and $O(1)$ working memory. Its cycle lengths are iid for
Gaussian input, and its long-stream mean span per output segment is
$21\zeta(3)E^2/(8\pi^2)$ asymptotically. Both segments of each cycle are
counted in this coefficient.
\end{enumerate}

The two asymptotic regimes are distinct. Brownian degeneration fixes the
error tolerance and observation horizon and refines the grid. The Gaussian
stationary result takes the segment-number limit at each fixed $E$ before
$E\to\infty$. It does not give a native segment-count law in a window whose
length varies with $E$. The repair has its own renewal count law, proved
for its specified stopping-time restarts.

\paragraph{Prior work and the role of the repair.}
Swing's failure to minimize segment count is already recognized by
Luo et al.~\cite{luo}, who distinguish continuous, disconnected and mixed
representations. Endpoint selection also has a substantial history:
Bristol's swinging-door patent discusses interior endpoints and the last
in-bounds observation~\cite{bristol}; classical sliding-window segmentation
grows sample windows until an error test fails~\cite{keogh}.
This windowing scheme does not by itself require the fitted segment
to pass through the first observed value. Elmeleegy et al. themselves
discuss recording the last observation before choosing the least-squares
recording~\cite[Section 3.2]{swing}. We do not claim that anchoring at a true
sample is a new principle. The contribution is a quantitative degeneration
result for inherited fitted recordings and a complete error and scaling
analysis of an explicitly specified repair. The bridge is included in both
the output and its segment cost.

Fresh Brownian affine approximation has a different, quadratic lifetime
scale. Chen~\cite[Theorem 1.2 and Lemma 2.1]{chenBrownian} evaluates the
anchored mean as $21\zeta(3)/\pi^2$ at unit tolerance, with quadratic
scaling at other tolerances. We use that result to evaluate the repair's
coefficient after proving the new algorithm's regeneration and mean-limit
conditions. We do not apply a fresh-start lifetime to the selected carry
state of native Swing. Free-knot stochastic approximation~\cite{free}
also permits local choices that do not impose this inherited recording.
Brownian-input parameter studies of swinging-door compression~\cite{ishii}
provide additional motivation for specifying the endpoint rule precisely.

The maximum-lag branch in the original algorithm fixes a candidate line
before feasibility is exhausted and then continues with a linear filter.
It has different dynamics and is excluded throughout. Our results do not
contradict Swing's pointwise error guarantee and do not claim that the
repair is optimal on every input.

Section~\ref{sec:algorithm} specifies the recording rule and proves
Brownian degeneration. Section~\ref{sec:means} derives the conditional
tail and explains the exceptional second-segment mean.
Section~\ref{sec:stationary} identifies the stationary coefficient and
shows why logarithmic mean order appears from the third segment onward.
Section~\ref{sec:repair} constructs the repair and proves its quadratic
mean-span law. Section~\ref{sec:evidence} evaluates the one unknown
coefficient and states the scope of the comparison. The appendices give
the probability estimates and explicit computational formulas.

\section{The first segment and the loss at restart}\label{sec:algorithm}

Let $f$ be sampled at times $0,h,2h,\ldots$. A recording $(a,q)$ is fixed
throughout the current interval. For the samples already accepted after $a$,
the feasible slopes form
\begin{equation}\label{eq:cone}
 I=[L,U]=\bigcap_{a<t_i\le b}
 \left[\frac{f(t_i)-q-\eps}{t_i-a},
       \frac{f(t_i)-q+\eps}{t_i-a}\right].
\end{equation}
The next sample is accepted exactly when its interval intersects $I$.
Otherwise the filter chooses $m\in I$, records
\begin{equation}\label{eq:record}
 (b,q')=(b,q+m(b-a)),
\end{equation}
and initializes a new interval at this recording using the rejected sample
at $b+h$. The rejected sample is excluded from the old line-selection loss.
The initial recording is $(0,f(0))$. An end-of-file output uses a feasible
line on the final prefix.

The scalar least-squares choice in~\cite{swing} is
\[
 m=\max\{L,\min\{A,U\}\},\qquad
 A=\frac{\sum_i(t_i-a)(f(t_i)-q)}{\sum_i(t_i-a)^2}.
\]
The Brownian degeneration results apply to any feasible slope choice measurable when rejection
is observed. We use the natural Brownian filtration for this causality
requirement. Independent auxiliary randomization can be included without
changing the conditional Brownian arguments.
The Gaussian second-segment, stationary and leading-coefficient results
below specifically use the constrained least-squares rule.

The first segment and the restarted segments solve different initial-value
problems. The first line passes through an observed value. After rejection,
the next line passes through the previous fitted recording, and the
rejected observation is already part of the new constraint set.
Translating coordinates preserves both this inherited error and the
selection of the carried observation.

\paragraph{The Gaussian scales at a glance.}
At unit sampling, let $L_{k,E}$ be the $k$th complete segment length
at tolerance $E$, and let $C_E=\lim_{k\to\infty}\E L_{k,E}$. Then
\begin{equation}\label{eq:scalesummary}
\begin{aligned}
 \E L_{1,E}&\sim\frac{21\zeta(3)}{\pi^2}E^2,
 &\E L_{2,E}&=\Theta(E),\\
 \E L_{k,E}&=\Theta_k(E\log E)\quad(k\ge3\text{ fixed}),
 &C_E&\sim\gamma_{\rm Swing}E\log E.
\end{aligned}
\end{equation}
See Theorems~\ref{thm:firstmean}, \ref{thm:secondsegment},
\ref{thm:transientconvergence} and~\ref{thm:leadingcoefficient}.
The first coefficient is the anchored Brownian mean in
\cite[Theorem 1.2]{chenBrownian}. The second mean is smaller than the
stationary mean by a logarithmic factor: immediate degeneration precedes
equilibration. The subsequent geometric bound concerns the means after
normalization by $E\log E$; it does not say that the second segment is
already near the stationary coefficient.

Write $b_{k,h}$ for the $k$th rejection recording time, $b_{0,h}=0$, and
\[
 J_{k,h}=b_{k,h}+h,\quad
 L_{k,h}=b_{k,h}-b_{k-1,h},\quad
 \eta_{k,h}=q_{k,h}-f(b_{k,h}),\quad e_{k,h}=\eta_{k,h}/\eps.
\]
The $J_{k,h}$ are stopping times; the $b_{k,h}$ generally are not.
For $k\ge1$ the exact identity
\begin{equation}\label{eq:durationidentity}
 J_{k+1,h}-J_{k,h}=L_{k+1,h}
\end{equation}
will allow a duration budget at the correct stopping times.
The process is defined on an infinite sampled path before restriction to a
finite horizon. Let $T_h=h\lfloor T/h\rfloor$ and
$R_h(T)=\#\{k:J_{k,h}\le T_h\}$. Initial and terminal outputs are excluded
from this count.

For a continuous path $f(0)=0$ and an initial fitted value $\eta$ with
$|\eta|<\eps$, define
\[
 L_f(t)=\sup_{0<s\le t}\frac{f(s)-\eta-\eps}{s},\qquad
 U_f(t)=\inf_{0<s\le t}\frac{f(s)-\eta+\eps}{s},\qquad
 \sigma=\inf\{t>0:L_f(t)=U_f(t)\}.
\]

\begin{lemma}\label{lem:pinch}
If $0<\sigma<\infty$, the feasible slope at $\sigma$ is a unique finite
number $m_*$ and $|f(\sigma)-\eta-m_*\sigma|=\eps$.
For standard Brownian motion, $\sigma$ is finite almost surely, is a
stopping time, and is the first feasibility-failure boundary. Starting
from either fitted value $B_\sigma\pm\eps$ admits no feasible line on
any positive-length future interval almost surely.
\end{lemma}
\begin{proof}
The optimized functions tend respectively to $-\infty$ and $+\infty$
at zero. Their extrema on a positive interval are attained away from zero,
and the running extrema are continuous. If the endpoint residual at
$\sigma$ were strictly interior, both active constraints would have been
attained earlier. The interval would therefore have collapsed earlier.

For Brownian motion, if all prefixes were feasible, nested compact slope
intervals would give a line feasible forever. The property $B_t/t\to0$
forces its slope to be zero, whereas Brownian motion does not stay in any
fixed bounded interval. At the finite stopping time $\sigma$, the strong
Markov property, time inversion, and the two-sided law of the iterated
logarithm give, for $W_s=B_{\sigma+s}-B_\sigma$,
\begin{equation}\label{eq:oscillation}
 \limsup_{s\downarrow0}W_s/s=+\infty,\qquad
 \liminf_{s\downarrow0}W_s/s=-\infty
 \quad\hbox{almost surely}.
\end{equation}
These standard Brownian facts are stated in~\cite{mp}, Theorems 1.9,
2.16, 5.1 and Remark 5.2. A line through $B_\sigma+\eps$ would require
$W_s\ge ms$, and one through $B_\sigma-\eps$ would require $W_s\le ms$.
Equation~\eqref{eq:oscillation} excludes both, simultaneously for all finite
$m$. It also prevents the unique old line from extending past $\sigma$.
\end{proof}

\subsection{Boundary concentration and alternating errors}\label{sec:boundary}

\begin{theorem}\label{thm:fixed}
Let $B$ be standard Brownian motion and let $\sigma,m_*$ refer to the
continuous initial problem with true anchor. Set
$S=(m_*\sigma-B_\sigma)/\eps$. For every fixed $K\ge1$,
\begin{equation}\label{eq:fixed}
 (b_{k,h},e_{k,h})_{k=1}^K
 \longrightarrow (\sigma,(-1)^{k-1}S)_{k=1}^K
 \quad\hbox{almost surely}.
\end{equation}
Moreover $\Prob(S=1)=\Prob(S=-1)=1/2$, and $S$ is independent of $\sigma$.
On $\{\sigma<T\}$, $R_h(T)\to\infty$, whereas $hR_h(T)\to0$ almost surely.
\end{theorem}
The first cone converges to the continuous pinch described above.
Every later fixed recording inherits a boundary error, so a positive
limiting interval would contradict Brownian oscillation at that boundary.
The rejection inequality then forces the error sign to reverse. Uniform
boundary estimates and the induction are proved in
Appendix~\ref{app:boundary}.

\begin{theorem}\label{thm:empirical}
On $\{\sigma<T\}$, almost surely,
\[
 \frac1{R_h(T)}\sum_{k=1}^{R_h(T)}\delta_{e_{k,h}}
 \ \Rightarrow\ \tfrac12\delta_{-1}+\tfrac12\delta_{+1}.
\]
On the same event, for every $a>0$, the proportion of completed intervals
of length at least $a$ tends to zero almost surely.
\end{theorem}
Uniform boundary concentration and alternation outside a bounded number
of long intervals give this empirical statement; see
Appendix~\ref{app:boundary}.

\subsection{The longest post-first interval}\label{sec:longest}

Define the maximum duration after the first recording, clipped at the
observation horizon, by
\begin{equation}\label{eq:maxdefinition}
 \mathcal M_h(T)=\max_{k\ge1}
   \bigl(\min\{b_{k+1,h},T_h\}-b_{k,h}\bigr)_+.
\end{equation}
It is zero if no rejection occurs before the horizon. It includes the
terminal incomplete interval, but excludes the initial interval.

\begin{theorem}\label{thm:max}
For every fixed $T,\eps>0$ and every causal feasible slope selector,
$\mathcal M_h(T)\to0$ in probability as $h\downarrow0$.
\end{theorem}

\begin{corollary}\label{cor:relative}
For each fixed observation horizon,
$\mathcal M_h(T)/L_{1,h}\to0$ in probability. The horizon may also be
replaced by $A L_{1,h}$ for any fixed $A>1$.
\end{corollary}
\begin{proof}
For $a,c>0$, the probability that the ratio exceeds $a$ is at most
$\Prob(\mathcal M_h(T)>ac)+\Prob(L_{1,h}<c)$.
Use Theorems~\ref{thm:fixed} and~\ref{thm:max}, then let $c\downarrow0$.
The random-horizon assertion follows by first localizing the tight family
$A L_{1,h}$ below a deterministic $T$.
\end{proof}

The fixed-number limit alone cannot control a segment whose index grows
with the grid. The proof instead works at the actual rejection stopping
times. It compares the probability of a macroscopic surviving segment
with a logarithmically larger truncated mean duration, then sums using
the total observation-time budget. This controls the entire post-first
horizon without treating the restart intervals as independent.
Appendix~\ref{app:longest} gives the conditional estimates and the proof.

\begin{corollary}\label{cor:timeweighted}
For every $a>0$, the total time covered by post-first clipped intervals
of duration at least $a$ tends to zero in probability and in $L^1$.
\end{corollary}
\begin{proof}
It is bounded by $T\ind_{\{\mathcal M_h(T)\ge a\}}$.
\end{proof}

\subsection{Comparison with continuous interpolation}\label{sec:comparison}

Let $M_h^*(T)$ be the minimum number of segments of a continuous PLA
fitting all samples up to $T_h$ within tolerance $\eps$, with knots at
sample times. Requiring the knots also to take their true sampled values
only makes the following comparison stronger.

\begin{theorem}\label{thm:ratio}
For every continuous path $f$ on $[0,T]$, $M_h^*(T)$ is uniformly
bounded as $h\downarrow0$. For Brownian input on $\{\sigma<T\}$,
\[
 \frac{M_h^{\rm Swing}(T)}{M_h^*(T)}\longrightarrow\infty
 \quad\hbox{almost surely},
\]
where the Swing count includes its terminal interval.
\end{theorem}
\begin{proof}
Choose $\delta>0$ with $2\omega_f(\delta)<\eps$. Partition the grid
into blocks of $\lfloor\delta/h\rfloor$ intervals and one possible
remainder. Interpolate the true values at block endpoints. Each chord
has error at most $2\omega_f(\delta)$, and each full block has duration
at least $\delta/2$ for small $h$. Thus $M_h^*(T)\le1+2T/\delta$.
The numerator differs from $R_h(T)$ by at most one and diverges by
Theorem~\ref{thm:fixed}.
\end{proof}

The comparison can also use an explicit online rule. Test every incoming
prefix chord through its first and last true sample. On failure, output
the preceding true sample as a shared knot and carry the current sample.
Every prefix of duration at most $\delta$ is feasible by the same modulus
bound. Each completed interval has duration at least $\delta-h$, so its
segment count is uniformly bounded. Both outputs are continuous and
use the same tolerance and data, without appealing to an index layout
or to disconnected-segment representation costs.

\section{How restart states determine the mean}\label{sec:means}

The pathwise degeneration does not by itself determine a mean duration. We first fix the restart state and identify the tail
that creates a logarithm. We then return to the state selected by the
actual first rejection, before considering stationarity in the next section.

\subsection{A frozen carry state}

After reflection, a fixed effective carry state $w\in\R$ imposes the
slope constraints
\[
 \left[\frac{w+W_{j-1}-2E}{j},\frac{w+W_{j-1}}{j}\right],\qquad j\ge1,
\]
where $W$ is the fresh Gaussian future after the rejected observation;
$j=1$ is the already read carry. Denote the maximal feasible prefix by
$L_E(w)$. Use the same future for all $E$ and define
\[
 A_w=\sup_{j\ge1}\frac{-(w+W_{j-1})}j.
\]
The strong law and downward oscillation give $0<A_w<\infty$ almost surely.

\begin{theorem}\label{thm:frozen}
The maximal feasible prefix length from this frozen state satisfies
\[
 \frac{L_E(w)}E\longrightarrow\frac2{A_w}
 \quad\hbox{almost surely}.
\]
If $w=-c<0$, then also
\[
 \frac{\E L_E(w)}E\longrightarrow\E\frac2{A_w}\le\frac2c.
\]
\end{theorem}
\begin{proof}
For fixed $x>0$ put $n=\lfloor Ex\rfloor$ and
$V_E=\max_{j\le n}|w+W_{j-1}|=o(E)$. Eventually the numerators below
are negative, and
\[
 \frac{-2E-V_E}n\le
 \max_{j\le n}\frac{w+W_{j-1}-2E}j
 \le\frac{-2E+V_E}n\longrightarrow-\frac2x.
\]
The running upper bound decreases to $-A_w$. Rational values of $x$
on either side of $2/A_w$, followed by monotonicity of prefix feasibility,
prove the pathwise limit.

If $w=-c<0$, the carried point forces $m\le-c$, and $A_w\ge c$.
Feasibility at index $n\ge2$ then implies
$W_{n-1}\le2E-c(n-1)$. Consequently
\[
 \Prob(L_E(w)\ge n)\le
 \Phi\left(\frac{2E-c(n-1)}{\sqrt{n-1}}\right)
 \le e^{-c^2(n-1)/8}\quad(n-1\ge4E/c).
\]
For $E\ge1$ this supplies an integrable exponential envelope for the
tails of $L_E(w)/E$ beyond a fixed level. Uniform integrability proves
the expectation assertion.
\end{proof}

This theorem explains the linear scale: an early upper slope approaches
$-A_w$, while the terminal lower slope is approximately $-2E/L$.
Their collision gives $L\approx2E/A_w$. It does not establish convergence
of the actual states $w_E$ generated by a stream, or justify averaging
the conditional result over those states.

\subsection{The logarithmic conditional mean}

For the effective state in Theorem~\ref{thm:frozen}, let
$g_E(w)=\E L_E(w)$. For $w\ge0$, define the first descending overshoot by
\[
 \tau_w=\inf\{n\ge1:w+W_n<0\},\qquad
 O_w=-(w+W_{\tau_w}),\qquad V(w)=w+\E O_w.
\]
The centered Gaussian walk crosses below zero almost surely.

\begin{theorem}\label{thm:conditionalmean}
For every fixed $w\ge0$,
\[
 \frac{g_E(w)}{E\log E}\longrightarrow4V(w).
\]
Convergence is uniform on every compact subset of $[0,\infty)$.
In particular $V(0)=1/\sqrt2$, so $g_E(0)\sim2\sqrt2 E\log E$.
There is an absolute constant $C$ such that, for $E\ge8$ and all real $w$,
\begin{equation}\label{eq:meanuniform}
 g_E(w)\le CE(1+w_+)\log(eE).
\end{equation}
\end{theorem}

The logarithm comes from a survival tail extending from the linear to
the quadratic scale. For fixed $w\ge0$ and $0<\delta<1/2$, the proof gives
\begin{equation}\label{eq:interiortail}
 \sup_{E^{1+\delta}\le n\le E^{2-\delta}}
 \left|\frac nE\Prob(L_E(w)\ge n)-4V(w)\right|\longrightarrow0.
\end{equation}
Summing this $E/n$ tail produces the factor $\log E$, although the random
length itself has the linear scaling limit in Theorem~\ref{thm:frozen}.
The ranges outside this interior interval are controlled separately in
Appendix~\ref{app:conditional}; an almost-sure frozen-state limit alone
would not justify the mean asymptotic.

The negative-state conclusion of Theorem~\ref{thm:frozen} and this theorem
concern the actual discrete carry model. They give different conditional
mean orders on the two sides of zero. The coefficient uses a first
overshoot, rather than a stationary overshoot distribution.

\subsection{Averaging the actual restart states}

\begin{theorem}\label{thm:firstmean}
For unit Gaussian increments and a true initial anchor,
\[
 \frac{L_{1,E}}{E^2}\Rightarrow\sigma,\qquad
 \frac{\E L_{1,E}}{E^2}\longrightarrow
 \E\sigma=\frac{21\zeta(3)}{\pi^2},
\]
where $\sigma$ is the unit-tolerance anchored Brownian lifetime.
\end{theorem}
The first segment ends when feasibility fails, before its slope is selected.
Exact Gaussian scaling and Theorem~\ref{thm:fixed} give the
distributional limit.  A shared disjoint-block tail estimate supplies
uniform integrability; see Appendix~\ref{app:firstmeanproof} and
Lemma~\ref{lem:threepointblock}.  The repair will use the same anchored
lifetime at half the tolerance.

The first restart changes this quadratic behavior. Its input is selected
by the preceding constrained least-squares fit, and the second segment
has linear mean order.

\begin{theorem}\label{thm:secondsegment}
For the actual constrained-SSE Swing stream with a true initial anchor,
unit Gaussian increments, carried rejection and no maximum lag, there
are absolute constants $c,C,E_0>0$ such that
\[
 cE\le\E L_{2,E}\le CE\qquad(E\ge E_0).
\]
Consequently $\E L_{2,E}/(E\log E)\to0$, and hence
$\E L_{2,E}/C_E\to0$.
\end{theorem}

To see why the logarithmic stationary contribution disappears, let $m_1$
be the first selected slope and let $R_1=E-|Y_{J_1}-q_1|$ be the margin
of its carried rejection. Every old feasible line rejects that sample,
which gives the pathwise inequality $(R_1)_+\le|m_1|$.
The first slope tends to zero, so this entrance has a vanishing positive
part. A uniform concentration estimate at the selected rejection controls
the layer near zero, while a negative carry cuts off the long survival
tail. The resulting conditional mean is at most
$CE[1+\log_+(1/|R_1|)]$ on the negative side, and this logarithmic
singularity is integrable under the entrance law. Appendix~\ref{app:firstrestart}
proves uniform moments for that law; Appendix~\ref{app:secondsegment}
uses them and the concentration estimate to prove the matching linear bounds.

Subsequent restart states are still selected and depend on $E$.
The next section first fixes $E$ and follows their distributions to
stationarity. It then shows how a finite number of actual restarts
produces positive margins and gives a geometric approach to the
stationary coefficient in the large-tolerance limit.

\section{From repeated restarts to the stationary coefficient}\label{sec:stationary}

We use iid $N(0,1)$ increments at unit sampling intervals and the
original constrained least-squares selector without maximum lag.
We first fix $E>0$ and take the segment-index limit of the restart chain.
We then increase the tolerance to identify its stationary coefficient
and relate it to the initial restart sequence.

At a rejection index $j$, let $q$ be the new recording at $j-1$ and put
\[
 x=y_j-q=z-\eta.
\]
This is the vertical difference between the already observed sample
and the new fitted anchor. It avoids any reflection convention when
the anchor error is zero. Let $W_0=0$ and let $W_n$ be the fresh Gaussian
partial sums after $j$. The sample values relative to the fitted anchor
are $D_i=x+W_{i-1}$ for $i\ge1$. Write
\[
 I_n=[a_n,b_n]=\bigcap_{i=1}^n
 \left[\frac{D_i-E}{i},\frac{D_i+E}{i}\right],\qquad
 N=\min\{n\ge1:I_{n+1}=\varnothing\}.
\]
The selected slope and the next state are exactly
\begin{equation}\label{eq:stateupdate}
 m=\max\left\{a_N,\min\left\{
 \frac{\sum_{i=1}^N iD_i}{\sum_{i=1}^N i^2},b_N\right\}\right\},
 \qquad x'=x+W_N-mN.
\end{equation}
Indeed, the next fitted anchor is the present anchor plus $mN$,
whereas the next rejected sample is the present anchor plus $x+W_N$.
The loss excludes the rejected sample. Both $N$ and $x'$ depend only
on $x$ and the fresh increments. The uniform tail bound below ensures
$N<\infty$ almost surely, so these measurable rules define a probability
kernel $P_E(x,A)=\Prob_x(x'\in A)$.

Put $h_E(x)=\E_xN$. Reflection symmetry gives
$h_E(x)=h_E(-x)=g_E(E+x)=g_E(E-x)$.
For the actual stream, read the first increment and set
$X_0=Z_1$ with law $\mu_0=N(0,1)$. Let $X_k$ be the state after
the $k$th completed segment. The initial observation time is 1;
subsequent observation times are rejection stopping times.
Their future increments are independent. Therefore
\begin{equation}\label{eq:meanrecursion}
 \mathcal L(X_k)=\mu_0P_E^k,\qquad
 \E L_{k,E}=\mu_0P_E^{k-1}h_E,\qquad k\ge1.
\end{equation}
This is a recursion of state distributions. It need not close as a
recursion of scalar segment means.

\begin{theorem}\label{thm:stationary}
For every fixed $E>0$, the kernel $P_E$ has a unique invariant
probability measure $\pi_E$. There are constants $0<\alpha_E<1$ and
$H_E<\infty$ such that, for every initial law $\mu$ and integer $j\ge0$,
\[
 \|\mu P_E^j-\pi_E\|_{\mathrm{TV}}
 \le(1-\alpha_E)^{\lfloor j/2\rfloor},\qquad
 1\le h_E(x)\le H_E\quad\hbox{for all }x\in\R.
\]
Here total variation means $\sup_A|\mu(A)-\pi_E(A)|$.
Consequently, the actual stream satisfies
\[
 \lim_{k\to\infty}\E L_{k,E}
 =C_E:=\int h_E\,d\pi_E\in[1,H_E],
\]
with error at most $H_E(1-\alpha_E)^{\lfloor(k-1)/2\rfloor}$.
\end{theorem}
Appendix~\ref{app:stationary} proves the result by a two-step common
component and a uniform duration bound. The common component establishes
existence and convergence at each fixed $E$; its probability is not
bounded below uniformly as $E$ grows.

The short-segment events in this proof establish a common probability
component; they impose no intervention on the input or the algorithm.
The displayed $\alpha_E$ can be very small, so the rate is not a
practical mixing-time estimate. Starting from $\pi_E$ makes every
segment mean equal to $C_E$. The actual initialization is $N(0,1)$,
and Theorem~\ref{thm:secondsegment} shows that the second mean is
asymptotically negligible relative to $C_E$.
A closed form for $C_E$ does not follow from this fixed-tolerance
theorem. The next two subsections determine its large-$E$ asymptotic without
interchanging the two limits in $k$ and $E$.

\subsection{Uniform control at large tolerance}

We now let $E\to\infty$ after taking the stationary law from
Theorem~\ref{thm:stationary}.  Write
\[
 R=E-|X|,\qquad X\sim\pi_E.
\]
The sign of $R$ indicates whether the carried observation lies inside
the tolerance band around the new fitted anchor.  Symmetry gives
$h_E(x)=g_E(E-|x|)$, and hence $C_E=\E_{\pi_E}g_E(R)$.

\begin{theorem}\label{thm:stationaryscale}
For every fixed $p>0$, the stationary boundary margins satisfy
\begin{equation}\label{eq:stationarymarginmoments}
 \sup_{E\ge E_0}\E_{\pi_E}|E-|X||^p<\infty
\end{equation}
for some $E_0>0$.  There is a constant $C>0$ such that
\begin{equation}\label{eq:stationarymeanscale}
 C_E\le CE\log E,\qquad E\ge E_0.
\end{equation}
\end{theorem}

The proof combines the conditional mean bound with control of the
selected overshoot and the constrained least-squares slope at actual
rejection times.  An absorption argument first bounds $C_E$ and then
gives uniform moments of the microscopic margin; see
Appendix~\ref{app:margins}.  These moments supply the integrability
needed to pass to the limiting margin law and the leading coefficient.
The next theorem identifies a positive coefficient and hence also
provides the matching lower bound.

\subsection{The limiting margin kernel and the leading coefficient}
\label{sec:leadingcoefficient}

Let $\nu_E$ be the invariant law of the reflected margin $R=E-|X|$.
Throughout this subsection the increments are independent $N(0,1)$ variables,
sampling has unit spacing, and Swing uses its original constrained least-squares
rule without a maximum-lag constraint.  The order of limits is to take the
invariant law at fixed $E$ and then let $E\to\infty$.

For $r\in\R$, let $W_0=0$ and let $W$ be a standard Gaussian random walk, and set
\begin{equation}\label{eq:leadingslope}
 A_r=\sup_{i\ge1}\frac{W_{i-1}-r}{i}.
\end{equation}
Oscillation and the strong law give $0<A_r<\infty$ almost surely, and the
supremum is attained at a finite index.  When $r<0$, the distribution may have
an atom at $A_r=-r$; this branch is retained.
For each fixed $a>0$, use another independent Gaussian walk $\widetilde W$ and
write
\[
 S_n^{(a)}=an-\widetilde W_n,\qquad
 \tau_a=\inf\{n\ge1:S_n^{(a)}>0\},\qquad H_a=S_{\tau_a}^{(a)}.
\]
Let $O_a$ have the equilibrium overshoot density
\begin{equation}\label{eq:equilibriumovershoot}
 q_a(o)=\frac{\Prob(H_a>o)}{\E H_a}\ind_{\{o\ge0\}}.
\end{equation}
Thus $O_a$ is the limiting overshoot above a distant level, not the first
ladder height $H_a$ or a single truncated Gaussian increment.  Define a kernel
by first sampling $A_r$ and then, conditional on $A_r=a$, sampling $O_a$:
\begin{equation}\label{eq:limitingmarginkernel}
 K(r,dy)=\operatorname{Law}(A_r-O_{A_r})(dy).
\end{equation}
Its density is
\begin{equation}\label{eq:limitingmargindensity}
 k(r,y)=\E\left[
 \frac{\Prob(H_{A_r}>A_r-y)}{\E H_{A_r}}
 \ind_{\{y<A_r\}}\right].
\end{equation}
Here each probability and expectation involving $H_{A_r}$ is evaluated at the
realized parameter $A_r$.  The slope distribution can also be computed from
\begin{equation}\label{eq:slopedistribution}
 \Prob(A_r\le a)=\ind_{\{r+a\ge0\}}
 \Prob\left(\sup_{n\ge0}(W_n-an)\le r+a\right),\qquad a>0.
\end{equation}

\begin{theorem}\label{thm:leadingcoefficient}
The kernel $K$ has a unique invariant probability $\nu_\infty$, which is
absolutely continuous and equivalent to Lebesgue measure.  Moreover,
\begin{equation}\label{eq:marginweaklimit}
 \nu_E\Rightarrow\nu_\infty,
\end{equation}
and the stationary segment mean satisfies
\begin{equation}\label{eq:leadingcoefficient}
 \lim_{E\to\infty}\frac{C_E}{E\log E}
 =\gamma_{\rm Swing}
 :=4\int_0^\infty V(r)\,\nu_\infty(dr)\in(0,\infty),
\end{equation}
where $V$ is the Gaussian first-overshoot function in
Theorem~\ref{thm:conditionalmean}.
\end{theorem}
In particular, $C_E=\Theta(E\log E)$ and $C_E/E\to\infty$.

The proof in Appendix~\ref{app:limit} has three steps. An early finite
prefix determines the limiting slope; a distant crossing then produces
its equilibrium overshoot. This yields the kernel $K$ with the actual
carry branch retained. Finally, uniform margin moments allow invariant
laws and conditional means to pass to the limit. Uniqueness of the
limiting invariant law selects one coefficient for the stationary regime.

The representation~\eqref{eq:leadingcoefficient} identifies the mean
far along the stream. To relate it to the initial loss of scale, we next
return to the actual sequence of restart states.

\subsection{The third segment and the approach to stationarity}

The first rejection leaves a margin whose positive part vanishes as
$E\to\infty$, which suppresses the second segment's logarithmic mean.
One transition of the limiting kernel $K$ produces a density that is
strictly positive on the whole real line. The third segment therefore
has a positive chance of starting with positive margin. The conditional
$E/n$ survival tail then contributes $E\log E$ to its mean, and the
same mechanism persists at each later fixed index.

\begin{theorem}\label{thm:transientconvergence}
For the actual constrained-SSE stream initialized at a true sample,
with unit Gaussian increments, carried rejection and no maximum lag,
\[
 \E L_{k,E}=\Theta_k(E\log E),\qquad k\ge3\text{ fixed}.
\]
There is a finite constant $B$, independent of $k$, such that
\begin{equation}\label{eq:transient-mainbound}
 \limsup_{E\to\infty}
 \left|\frac{\E L_{k,E}}{E\log E}-\gamma_{\rm Swing}\right|
 \le B\left(\frac34\right)^{k-3},\qquad k\ge3.
\end{equation}
In particular, the left-hand side tends to zero as $k\to\infty$.
\end{theorem}

The proof needs uniform moments for the first entrance, rather than an
explicit entrance distribution. Along any convergent subsequence of
entrance laws, compact-uniform kernel convergence propagates the states.
The conditional mean theorem and moment bounds then pass their means
to the limit. Factoring $K$ into a margin-to-slope step and a
slope-to-margin step gives a slope kernel $T$ with a $3/4$ Wasserstein
contraction and a Lipschitz mean reward. This contraction gives the
same bound for every subsequential entrance law, proving
\eqref{eq:transient-mainbound}. See Appendix~\ref{app:transientconvergence}
for the mean passage and Appendix~\ref{sec:newconstant} for $T$.

The theorem fixes $k$ before taking $E\to\infty$. It quantifies the
approach of normalized means to the stationary coefficient in that
order of limits; it does not assert uniformity in a growing segment
index or a mixing time at a specified finite tolerance.

Native Swing thus loses the quadratic scale at its first restart, and
later averaging recovers only a logarithmic factor. The next section
changes the recording rule so that every new cycle starts at a true
sample and the quadratic scale is restored.

\section{A repair with continuous output and quadratic mean span}
\label{sec:repair}

The degeneration is a consequence of the recording convention in
Section~\ref{sec:algorithm}.  A conservative modification suffices to
restore the quadratic scale while keeping a continuous reconstruction,
the same final error tolerance, and constant auxiliary memory.  The
modification reserves half of the error budget for replacing a feasible
line by a chord.  It also encodes the rejected sample in a one-step
bridge, so that the next cycle starts at a genuine stopping time.
The extra bridge costs one output segment per cycle.

For reference, let
\[
 D_{\rm an}(f,t)=\inf_{a\in\R}\sup_{0\le u\le t}|f(u)-au|,
 \qquad
 \tau_{\rm an}=\inf\{t>0:D_{\rm an}(W,t)>1\},
\]
where $W$ is standard Brownian motion.  Chen
\cite[Theorem~1.2 and Lemma~2.1]{chenBrownian} gives
\begin{equation}\label{eq:repair-anchored-constant}
 \kappa_{\rm an}:=\E\tau_{\rm an}
 =\frac{21\zeta(3)}{\pi^2},
 \qquad
 \tau_{\rm an}(\eps,\upsilon)
 \overset d=\frac{\eps^2}{\upsilon^2}\tau_{\rm an}.
\end{equation}
Here the latter lifetime uses tolerance $\eps$ for
$dt+\upsilon W_t$, $\upsilon>0$, and a line through the true initial point.
These are fresh anchored lifetimes.  The cited result does not give the
lifetime of a native Swing restart, nor the segment count of an optimal
continuous PLA.

\paragraph{The rule.}
For samples $(t_i,y_i)$ with strictly increasing times, begin each cycle
at a true sample $(t_s,y_s)$ and maintain
\begin{equation}\label{eq:repair-cone}
 I_j=\bigcap_{i=s+1}^{j}
 \left[\frac{y_i-y_s-\eps/2}{t_i-t_s},
       \frac{y_i-y_s+\eps/2}{t_i-t_s}\right].
\end{equation}
On the first empty intersection at index $j$, output the true samples
$(t_{j-1},y_{j-1})$ and $(t_j,y_j)$ and restart from the latter.
Thus the completed cycle consists of the chord from $s$ to $j-1$ and
the bridge from $j-1$ to $j$.  The rejected sample is represented by the
bridge endpoint and becomes the next true anchor; no sample is discarded.
At end of input, output the last true sample if a nonempty remainder
exists.  The algorithm does not merge adjacent collinear segments.
Updating the two interval endpoints in \eqref{eq:repair-cone} takes
$O(1)$ arithmetic operations and comparisons per sample, in the
unit-cost model for real arithmetic.  The anchor, preceding sample,
and interval endpoints occupy $O(1)$ real-valued registers, excluding output.
No least-squares accumulator or stored prefix is required.

\begin{theorem}[Deterministic error and dense-sampling bound]
\label{thm:repair-error}
The rule produces a continuous piecewise-linear reconstruction whose
error is at most $\eps$ against the piecewise-linear interpolation of
the input samples.  For samples of any fixed $f\in C[0,T]$, its segment
count $M_h^{\rm repair}(T)$ is uniformly bounded as the sampling mesh
tends to zero.  In particular, on the Brownian event $\{\sigma<T\}$
from Theorem~\ref{thm:ratio},
\[
 \frac{M_h^{\rm Swing}(T)}{M_h^{\rm repair}(T)}
 \longrightarrow\infty\qquad\hbox{almost surely}.
\]
\end{theorem}
\begin{proof}
At rejection, take any line $\ell$ through $(t_s,y_s)$ with slope in
the preceding nonempty interval.  It has sample error at most $\eps/2$
through $j-1$.  If $c$ is the chord through the true endpoints $s$ and
$j-1$, then, throughout that time interval,
\[
 |c(t)-\ell(t)|
 =\frac{t-t_s}{t_{j-1}-t_s}
    |y_{j-1}-\ell(t_{j-1})|\le\eps/2.
\]
The chord therefore has sample error at most $\eps$.  On every original
grid interval, both the chord and the sample interpolation are affine,
so their difference is bounded by its endpoint errors.  The bridge
coincides with the sample interpolation on its single grid interval.
The same chord argument handles the terminal remainder.  A single
new sample always admits an anchored line, so a rejected cycle has at
least two grid intervals and both of its output segments have positive
length.

Choose $\delta>0$ such that the modulus of continuity of $f$ satisfies
$\omega_f(\delta)<\eps/2$.  Every prefix of duration at most $\delta$
admits the horizontal line through its true initial sample in
\eqref{eq:repair-cone}.  A cycle ending at a rejected sample consequently
has duration greater than $\delta$.  There are at most $T/\delta$
completed cycles and one terminal remainder, giving
\[
 M_h^{\rm repair}(T)\le 2T/\delta+1.
\]
This bound holds for every sampling grid.  On $\{\sigma<T\}$ the
native count diverges almost surely by the proof of
Theorem~\ref{thm:ratio}; division by this uniform bound proves the last
claim.
\end{proof}

The error statement concerns observed data and their linear
interpolation, as in the native sampled problem.  It does not assert a
deterministic bound on an unobserved Brownian bridge at a fixed mesh.
The repair is not claimed to minimize segment count on every input.

\begin{theorem}[Gaussian renewal law and quadratic mean span]
\label{thm:repair-gaussian}
Let $Y_0=0$ and let $Y_n-Y_{n-1}$ be iid standard normal variables,
sampled at unit times.  Use final tolerance $E>0$ and define
\begin{equation}\label{eq:repair-cycle}
 K_E=\inf\left\{n\ge1:
   \inf_{a\in\R}\max_{0\le i\le n}|Y_i-ai|>E/2\right\}.
\end{equation}
The successive cycle lengths of the repair are iid copies of $K_E$,
with finite mean.  If $M_E^{\rm repair}(N)$ counts its output segments
through sample $N$, including its terminal remainder, then
\begin{equation}\label{eq:repair-renewal}
 \frac{M_E^{\rm repair}(N)}N\longrightarrow\frac2{\E K_E}
 \quad\hbox{almost surely as }N\to\infty.
\end{equation}
Moreover,
\begin{equation}\label{eq:repair-mean}
 \frac{K_E}{E^2}\Rightarrow\frac{\tau_{\rm an}}4,
 \qquad
 C_E^{\rm repair}:=\frac{\E K_E}{2}
 \sim\frac{21\zeta(3)}{8\pi^2}E^2.
\end{equation}
Thus the long-stream limit is taken at fixed $E$, before $E\to\infty$.
\end{theorem}
Here $C_E^{\rm repair}$ is the limit of total span divided by output
segment count. From a cycle boundary, segment lengths alternate between
$K_E-1$ and $1$, so it is not a claim that their individual expectations
converge along segment number. Equivalently, it is the segment mean after
randomizing the initial phase uniformly between the two cycle positions.
\begin{proof}
The first rejection $K_E$ is a stopping time for the increment
filtration.  On the same input path, the native first segment at
half-tolerance has the same feasible prefixes, so exactly
\begin{equation}\label{eq:repair-firstduration}
 K_E=L_{1,E/2}+1.
\end{equation}
Its first slope is chosen only after rejection and does not affect
this identity.  Lemma~\ref{lem:threepointblock}, conditioned on the
first observation and used at tolerance $E/2$, gives $\E K_E<\infty$.

The repair starts its next cycle at the observed value $Y_{K_E}$.
The increments strictly after this stopping time are independent of
the completed cycle and have the original Gaussian law, proving the
iid cycle statement.  Each complete cycle produces two segments.
The strong law for the cycle lengths and inversion of their partial
sums therefore give~\eqref{eq:repair-renewal}; a terminal remainder
adds at most one segment to twice the completed-cycle count.

Theorem~\ref{thm:firstmean} supplies the weak limit and the limiting
expectation of the first duration divided by $(E/2)^2$.
Its limiting unit-tolerance lifetime $\sigma$ is the $\tau_{\rm an}$
defined above.  Thus~\eqref{eq:repair-firstduration} gives
\[
 \frac{K_E}{E^2}\Rightarrow\frac{\tau_{\rm an}}4,
 \qquad
 \frac{\E K_E}{E^2}\longrightarrow\frac{\kappa_{\rm an}}4.
\]
The mean passage uses the uniform integrability already proved for
the sampled anchored lifetime in Appendix~\ref{app:firstmeanproof}.
Dividing by two output segments per cycle proves~\eqref{eq:repair-mean}.
\end{proof}

Combining \eqref{eq:repair-mean} with the native stationary mean in
Theorem~\ref{thm:leadingcoefficient} gives the comparison
\begin{equation}\label{eq:repair-native-ratio}
 \frac{C_E^{\rm repair}}{C_E^{\rm Swing}}
 \sim\frac{21\zeta(3)}{8\pi^2\gamma_{\rm Swing}}
       \frac E{\log E}.
\end{equation}
The denominator here is the native stationary mean from that theorem;
\eqref{eq:repair-renewal} is a separate pathwise statement for the
repaired algorithm.  The repair establishes the quadratic order with
an explicit coefficient; optimizing its conservative half-budget and
bridge overhead is outside the present claim.

\section{Evaluating the coefficient and scope}\label{sec:evidence}

\subsection{One numerical coefficient}

The scale comparison leaves one unknown number to evaluate:
$\gamma_{\rm Swing}$. Its margin representation in
\eqref{eq:leadingcoefficient} can be reduced to a positive-slope kernel
$T$ and a reward $h$, with $\gamma_{\rm Swing}=\lambda h$ and
$\lambda T=\lambda$. Appendix~\ref{sec:newconstant} gives explicit
Gaussian formulas for both objects, including the atoms in the maximum
distribution. Theorem~\ref{thm:constantcontraction} gives a convergent
construction from the chosen boundary input zero:
\[
 q_m=(T^mh)(0),\qquad
 q_{2m}\uparrow\gamma_{\rm Swing},\qquad
 q_{2m+1}\downarrow\gamma_{\rm Swing},\qquad
 |q_m-\gamma_{\rm Swing}|\le L_hM_0(3/4)^m,
\]
where $L_h,M_0$ are explicit bounds given in that appendix. This
construction evaluates the stationary coefficient without estimating
the distribution selected by the first native rejection.

Deterministic quadrature of the limiting operator gives
\[
 \boxed{\gamma_{\rm Swing}\approx1.8120703.}
\]
Two evaluations, one using the mean reward and one using the boundary
identity, converge to the same value under mesh
refinement. A separate discretization of the Lindley equation checks
the maximum distribution. Appendix~\ref{sec:constantnumerics} records
the finite-grid output-CDF estimate, meshes, domain sensitivity and
operator mass defects. The decimal
is numerical evidence, not a certified interval: the exact iteration
bound above does not include quadrature, truncation or rounding errors.
The first-segment and repair coefficients are already explicit through
the anchored Brownian constant and require no additional numerical
constant.

\subsection{What the comparison establishes}

The results describe three successive events: the fitted recording
reaches an error boundary; the inherited margin lowers the mean from
$E^2$ to $E$ at the next segment; later positive margins create an
$E\log E$ mean. Stationary averaging changes the coefficient, while
true-sample regeneration restores the power $E^2$. This distinction
explains both the immediate loss and the effectiveness of the repair.

The mean statements have specific orders of limits. The native stationary
law is taken at fixed $E$ before increasing the tolerance, and the
transient bound fixes the segment index before that limit. The repair
has a separate, proved renewal count law. We do not infer a native
finite-window count law from its stationary segment mean.

Finite-grid checks and exact rational rescans accompany the computation
records. They check feasible cones, the constrained least-squares choice,
carried rejection and the repair's chord errors. Finite grids can still
have long post-first intervals; the Brownian longest-segment theorem
is a convergence-in-probability statement without a quantitative grid
rate. These checks support the implementations, while the probability
and error guarantees follow from the proofs above.

The repair preserves the prescribed sampled-data error and continuous
output with constant work and memory per observation. Including its
bridge cost, it restores a quadratic mean span from the native
stationary $E\log E$ scale.

\appendix
\section{Boundary concentration and fixed-recording limits}\label{app:boundary}

\begin{lemma}\label{lem:uniform}
For any fixed continuous path on a compact interval, the least number of
new samples in a completed interval tends to infinity as $h\downarrow0$.
If there is no completed interval, take this least count to be $+\infty$;
empty suprema below are taken to be zero.
Uniformly over its rejection recordings,
\begin{equation}\label{eq:uniform}
 \sup|m|h\longrightarrow0,\qquad
 \sup(\eps-|\eta_{k,h}|)\longrightarrow0.
\end{equation}
\end{lemma}
\begin{proof}
For any fixed integer $n$, from an anchor with error $|\eta|\le\eps$
consider the line with values $f(a)+\eta(1-i/n)$ at $a+ih$,
$0\le i\le n$. If the path modulus satisfies $\omega_f(nh)<\eps/n$,
the line fits all $n$ new samples. This holds uniformly over anchors.
A path may be extended to a slightly larger compact interval to avoid an
end-of-file restriction. Thus the least accepted new-point count $k_h$
tends to infinity. Endpoint feasibility gives
\[
 |m|h\le\frac{2\eps+\osc(f)}{k_h}.
\]
For an upper rejection, put $r_i=f(t_i)-\ell(t_i)$ for the selected old
line. Since $r_j>\eps$ and $r_{j-1}\le\eps$,
\begin{equation}\label{eq:signedbound}
 0\le\eps+\eta_{k,h}=\eps-r_{j-1}
 <\Delta f_j-mh.
\end{equation}
The lower-rejection inequality is its reflection. Consequently
$0\le\eps-|\eta_{k,h}|\le\omega_f(h)+|m|h$, uniformly.
\end{proof}

\begin{lemma}\label{lem:direction}
Suppose the path displacement from an interval's starting point, including
its rejected sample, is everywhere less than $\eps$ in absolute value.
If the starting recording error is nonnegative, rejection below the old
slope interval is impossible. If it is nonpositive, rejection above that
interval is impossible.
\end{lemma}
\begin{proof}
Write the starting error as $\eta\ge0$, the relative samples as
$(s_i,x_i)$, the rejected index as $j$, and $d=\max_{i\le j}|x_i|<\eps$.
For $i<j$, negativity of $d-\eta-\eps$ gives
\[
 \frac{x_i-\eta-\eps}{s_i}\le
 \frac{d-\eta-\eps}{s_j},\qquad
 \frac{x_j-\eta+\eps}{s_j}\ge
 \frac{-d-\eta+\eps}{s_j}.
\]
The new upper bound exceeds every old lower bound by at least
$2(\eps-d)/s_j>0$. Reflection proves the other assertion.
\end{proof}

\begin{proof}[Proof of Theorem~\ref{thm:fixed}]
Before $\sigma$ the continuous slope interval has nonempty interior.
After $\sigma$, two strictly separated constraints witness infeasibility.
Dense samples eventually detect those constraints, so $b_{1,h}\to\sigma$.
Any first-interval chosen slope is bounded using a sample near $\sigma/2$.
Every subsequential slope limit fits $[0,\sigma]$ and must equal $m_*$.
Thus $e_{1,h}\to S$.

Suppose $b_{k-1,h}\to\sigma$. If the next interval extended to
$\sigma+u$ along a subsequence, an accepted sample near $\sigma+u/2$
would bound its slope. Lemma~\ref{lem:uniform} allows a further subsequence
whose starting error tends to $s\eps$, $s\in\{-1,1\}$.
Continuity and dense sampling would then give a finite-slope line through
$B_\sigma+s\eps$ fitting a positive future interval, contradicting
Lemma~\ref{lem:pinch}. This also rules out an interval that never ends.
Hence $b_{k,h}\to\sigma$.

The entire $k$th interval including its rejection now has vanishing path
oscillation for every fixed $k\ge2$. Lemma~\ref{lem:direction} and
\eqref{eq:signedbound} force its terminal error to the opposite boundary.
Induction proves~\eqref{eq:fixed}.
Reflection $B\mapsto-B$ fixes $\sigma$ and changes $S$ to $-S$, proving
fairness and independence. Every fixed number of recordings eventually
precedes $T$ when $\sigma<T$, so their count diverges.
Finally $hR_h(T)\le T/k_h$ with $k_h$ from Lemma~\ref{lem:uniform}.
\end{proof}

\begin{proof}[Proof of Theorem~\ref{thm:empirical}]
Choose $\rho>0$ with $\omega_B(\rho)<\eps$. For small $h$, uniform
boundary concentration makes all recording errors nonzero.
By Lemma~\ref{lem:direction} and~\eqref{eq:signedbound}, successive
recording errors of the same sign require an interval, including its
rejection, longer than $\rho$. Its accepted duration exceeds $\rho-h$.
Disjoint interiors imply at most $2T/\rho$ such transitions for small $h$.
Splitting the sign sequence at these transitions gives alternating runs.
The positive and negative counts therefore differ by at most $1+2T/\rho$.
Divide by $R_h(T)$ and use uniform boundary concentration.
The duration assertion follows from
$\#\{k:L_{k,h}\ge a,\ J_{k,h}\le T_h\}\le T/a$.
\end{proof}

\section{Controlling all post-first Brownian intervals}\label{app:longest}

The fixed-number argument cannot establish Theorem~\ref{thm:max}. An interval
selected from the entire stream can have a growing index. A limit of its
starting times need not be a stopping time, and Brownian paths do have
one-sided supporting points at exceptional times. We instead estimate
conditional probabilities at the discrete rejection stopping times.
Brownian scaling reduces the proof to $\eps=1$.

At a recording time $b$ with error $\eta$ and carry increment
$z=B_{b+h}-B_b$, choose $s=1$ when $\eta\ge0$ and $s=-1$ otherwise.
Reflect the data when $s=-1$.
At the rejection stopping time $r=b+h$, put
\[
 w=1-|\eta|+sz.
\]
Future reflected increments $W$ are conditionally standard Brownian
motion. A candidate slope $m$ must satisfy
\begin{equation}\label{eq:state}
 w+W_u-m(h+u)\in[0,2]
\end{equation}
at every future accepted grid point. In particular the carry interval is
$[(w-2)/h,w/h]$.
Let $H(a,M)$ be the existence of a line with $|m|\le M$ fitting through
the anchor-relative grid time $h\lceil a/h\rceil$.

\begin{lemma}\label{lem:ratioestimates}
Fix $a,M>0$. For states $w\ge-Mh$, write $d=w_++\sqrt h$.
There are positive constants $t_0,c,C$, depending at most on $M$, such
that, for small $h,d$ with $Cd<t_0$,
\begin{align}
 \Prob(H(a,M)\mid\F_r)
 &\le C_{a,M}\bigl(d+\sqrt{h\log(1/h)}\bigr)+C_a h^3,
 \label{eq:tail}\\
 \E[L\wedge t_0\mid\F_r]
 &\ge cd\log\frac{t_0}{Cd},\label{eq:meanlower}
\end{align}
where $L$ is the complete duration of the interval starting at $b$.
\end{lemma}
\begin{proof}
For the upper bound, $m\ge-M$ and~\eqref{eq:state} imply grid nonnegativity
of $x+W_u+Mu$, where $x=w+Mh\ge0$, for at least $a/3$ time units when
$h$ is small. Brownian motion differs from its piecewise-linear grid
interpolant by at most $\rho_h=C\sqrt{h\log(1/h)}$ except on an event of
probability $C_a h^3$. This follows by bounding each interval's difference
by twice its maximum increment, then applying reflection and a union bound.
For $q\ge a/3$, the drifted reflection formula is
\[
 \Prob\!\left(\min_{u\le q}(W_u+Mu)\ge-y\right)
 =\Phi\!\left(\frac{y+Mq}{\sqrt q}\right)
 -e^{-2My}\Phi\!\left(\frac{-y+Mq}{\sqrt q}\right).
\]
It vanishes at $y=0$ and has derivative at most
$2/\sqrt{2\pi q}+2M$ for $y\ge0$. Take $y=x+\rho_h$.

For the lower bound choose $t_0$ small, with $t_0\le1/(2M)$.
For each target time $Cd\le t\le t_0$, use the candidate $m=-1/t$.
Its carry coordinate is $w+h/t\in[0,2]$.
Require the first fresh increment after $r$ to lie in
$[\sqrt h,2\sqrt h]$, an event of fixed probability
$p_0=\Phi(2)-\Phi(1)$. At the next sample the coordinate is
\[
 v=w+W_h+2h/t.
\]
Uniformly, $d/2\le v\le4d$ and $v/t\le1/2$ if $C$ is sufficiently
large and $h,d$ are small. Subsequently a Brownian motion with drift
$1/t$ starting at $v$ need only remain in $(0,2)$.

For that process $X$, let $\tau_0,\tau_2$ be its boundary hitting times.
The probability of never hitting zero is $1-e^{-2v/t}\ge c_1v/t$.
Set
\[
 f_t(x)=e^{-x/t}\sinh(2x/t),\qquad \lambda_t=3/(2t^2).
\]
Direct differentiation gives
$\tfrac12f_t''+t^{-1}f_t'=\lambda_t f_t$.
Let $\tau=\tau_0\wedge\tau_2$. The stopped process
$M_u=e^{-\lambda_t(u\wedge\tau)}f_t(X_{u\wedge\tau})$
is a bounded nonnegative martingale, so $\E_v M_t=f_t(v)$.
On $\{\tau_2\le t,\tau_2<\tau_0\}$,
$M_t\ge e^{-\lambda_t t}f_t(2)$. Since $v/t\le1/2$,
\[
 \Prob_v(\tau_2\le t,\tau_2<\tau_0)
 \le e^{\lambda_t t}\frac{f_t(v)}{f_t(2)}
 \le C_1(v/t)e^{-1/(2t)}.
\]
After reducing $t_0$, subtraction from the probability of never hitting
zero yields
\[
 \Prob_v(\tau_0>t,\tau_2>t)\ge c_2v/t.
\]
The upper-boundary error retains the factor $v/t$; a merely small
additive error would not suffice. Thus the candidate fits at least through
$h\lfloor t/h\rfloor$, and maximal-prefix acceptance gives
$\Prob(L\ge t-h\mid\F_r)\ge c_3d/t$.
Integration over $Cd\le t\le t_0$ proves~\eqref{eq:meanlower}.
The slope depends on the target time. No independence is assigned to the carry.
\end{proof}

\begin{proof}[Proof of Theorem~\ref{thm:max}]
Fix $a>0$ and a path-height cutoff $K$. On
$\{\sup_{u\le T}|B_u|\le K\}$, an interval clipped to duration at least
$a$ admits a feasible line with $|m|\le M=(2K+2)/a$, by its endpoint
errors. Its carry necessarily satisfies $w\ge-Mh$.

For small fixed $\gamma>0$, let $G_h(\gamma)$ say that all recordings
whose rejections precede $T$ satisfy $1-|\eta|\le\gamma$ and
$|z|\le\gamma$. Lemma~\ref{lem:uniform} and path continuity give
$\Prob(G_h(\gamma)^c)\to0$.
Let $A_{k,h}\in\F_{J_{k,h}}$ express these state restrictions together
with $J_{k,h}\le T$ and $w\ge-Mh$.
On $A_{k,h}$, $\sqrt h\le d\le2\gamma+\sqrt h$.

Divide~\eqref{eq:tail}, apart from its $O(h^3)$ term,
by~\eqref{eq:meanlower}. The resulting ratio is at most
\[
 \beta_h(\gamma)=C_{a,M}
 \sup_{\sqrt h\le d\le2\gamma+\sqrt h}
 \frac{d+\sqrt{h\log(1/h)}}{d\log(t_0/(Cd))}.
\]
Choose $\gamma$ small enough that the denominators are positive.
For $d\le h^{1/4}$ the ratio tends to zero at least as fast as
the inverse square root of $\log(1/h)$.
For larger $d$, the second numerator term is negligible relative to $d$.
It follows that
\begin{equation}\label{eq:betalim}
 \lim_{\gamma\downarrow0}\limsup_{h\downarrow0}\beta_h(\gamma)=0.
\end{equation}

The exact duration identity~\eqref{eq:durationidentity} gives the
pathwise bound
\begin{equation}\label{eq:budget}
 \sum_{k:J_{k,h}\le T}(L_{k+1,h}\wedge t_0)\le T+t_0.
\end{equation}
There is at most one crossing interval. A union bound over
$A_{k,h}\cap H_k(a,M)$, the tower property, and~\eqref{eq:budget} give
\begin{align*}
 \Prob(\mathcal M_h(T)\ge a)
 &\le\Prob(\sup_{u\le T}|B_u|>K)+\Prob(G_h(\gamma)^c)\\
 &\quad+\beta_h(\gamma)(T+t_0)+O_{a,T}(h^2).
\end{align*}
There are at most $O(T/h)$ possible rejections for the additive errors.
Let $h\downarrow0$, then $\gamma\downarrow0$, then $K\uparrow\infty$.
The argument also covers a final interval cut off by $T_h$.
\end{proof}

\section{Conditional means and the sampled anchored lifetime}\label{app:conditional}

\subsection{A uniform block tail}

\begin{lemma}[Three-point blocks]\label{lem:threepointblock}
Fix a tolerance $E>0$, an input history and its carried value $D_1=x$
relative to a fitted anchor at index zero.  Let
$D_i=x+W_{i-1}$ for $i\ge1$, where $W$ has fresh unit Gaussian increments,
and let $N$ be the maximal feasible prefix length.  Put
\[
 r_E=\lceil8E^2\rceil,\qquad
 p_* =\Prob(|Z|\le1)<1,\qquad Z\sim N(0,1).
\]
Let $\mathcal H_n$ be generated by the input history and observations
$D_1,\ldots,D_n$.
For all integers $n\ge1$ and $b\ge0$,
\begin{equation}\label{eq:threepointconditional}
 \Prob(N\ge n+2r_Eb\mid\mathcal H_n)
 \le\ind_{\{N\ge n\}}p_*^b.
\end{equation}
In particular, uniformly over the input history and $x$,
\begin{equation}\label{eq:threepointmean}
 \Prob(N\ge1+2r_Eb)\le p_*^b,\qquad
 \E N\le H_E:=\frac{2r_E}{1-p_*}.
\end{equation}
\end{lemma}
\begin{proof}
An affine line fitting three equally spaced accepted points within $E$
requires
\[
 |D_{i+2r_E}-2D_{i+r_E}+D_i|\le4E.
\]
For $i=n+2r_Ea$, $0\le a<b$, these second differences use disjoint
increments after $D_n$.  They are independent $N(0,2r_E)$ variables,
independent of $\mathcal H_n$; the known value $D_n$ cancels.
Since $4E/\sqrt{2r_E}\le1$, survival requires $b$ independent events
of probability at most $p_*$.  The event $\{N\ge n\}$ is
$\mathcal H_n$-measurable, proving~\eqref{eq:threepointconditional}.
A single carried point is always feasible, and summing the resulting
tail in blocks of $2r_E$ proves~\eqref{eq:threepointmean}.
\end{proof}

\subsection{A discrete survival estimate}

Write $Q_n^a(x)=\Prob(x+W_j+aj\ge0,\ 1\le j\le n)$ and
$Q_\infty^a(x)=\lim_n Q_n^a(x)$.

\begin{lemma}\label{lem:ballot}
For $x\ge0$, $0\le a\le1/2$, and integer $n\ge1$,
\[
 Q_n^a(x)\le C(x+1)(a+n^{-1/2}).
\]
\end{lemma}
\begin{proof}
If $G$ is standard normal, the conditional excess $G-r$ given $G>r$
has uniformly bounded moments of orders one and two, and uniformly
bounded $\E[(G-r)e^{G-r}]$, for $r\ge-1/2$.
For $r\ge0$, its upper tail is at most $e^{-u^2/2}$, by shifting the
Gaussian density in the tail integral. For $-1/2\le r<0$, bound the
conditional tail by $2\bar\Phi(u-1/2)$. In particular an exit overshoot
$O$ satisfies $\E e^{2aO}\le1+Ca$ for $a\le1/2$.

Put $X_j=x+W_j+aj$ and stop at the first exit $\tau$ from $[0,R]$,
where $R\ge1$, $aR\le1/2$, and $x<R$.
For $f(u)=u(R-u)$, the conditional one-step change before exit is
$a(R-2X_j)-1-a^2\le-1/2$. At exit,
$f(X_\tau)=-RO-O^2$ for either exit direction.
Stopping first at $\tau\wedge m$, using the uniform overshoot moments,
and then letting $m\to\infty$ gives
$\E\tau\le CR(x+1)$.

Let $p$ be the upper-exit probability. For $a>0$, optional stopping
of $e^{-2aX_j}$ is justified by the uniform exponential moments of the
lower-exit overshoot. It gives
\[
 e^{-2ax}\le(1-p)(1+Ca)+pe^{-2aR},\qquad
 p\le\frac{1+Ca-e^{-2ax}}{1+Ca-e^{-2aR}}
 \le C(x+1)/R.
\]
For $a=0$, stop $X_j$ instead and use $x\ge pR-(1-p)C$.
Choose $R=\min\{\sqrt n,1/(2a)\}$, with the second term infinite
when $a=0$. If $x\ge R$, the desired bound is trivial after increasing
$C$. Otherwise survival to $n$ requires upper exit or $\tau>n$.
Thus it is at most $C(x+1)(R^{-1}+R/n)$, which proves the lemma.
\end{proof}

\subsection{Small drift and a uniform duration bound}

For $a>0$ and $x\ge0$, let $\tau_{a,x}$ be the first negative position of
$x+W_n+an$, and $O_{a,x}$ its overshoot on hitting. Conditional on a
particular exit history and exit, its tail is bounded by $e^{-u^2/2}$.
The exponential martingale, stopped at exit or a finite time, therefore gives
\begin{equation}\label{eq:overshootidentity}
 e^{-2ax}=\E[e^{2aO_{a,x}};\tau_{a,x}<\infty].
\end{equation}
On paths not hitting, its value tends to zero by the strong law.
As $a\downarrow0$ and $x\to w\ge0$, $x\ge0$, the hitting prefix and
overshoot converge almost surely to those at zero drift. Indeed, the
latter hitting time is finite, and none of its finitely many noninitial
positions equals zero almost surely. Uniform overshoot moments give
convergence in mean. Expanding~\eqref{eq:overshootidentity} yields
\begin{equation}\label{eq:smalldrift}
 Q_\infty^a(w+a)/(2a)\longrightarrow V(w).
\end{equation}
The same coupling proves continuity of $V$. Monotonicity in the starting
position and a finite partition of a compact interval prove uniformity
of~\eqref{eq:smalldrift} there.

The probability of eventual ruin from $y\ge0$ is at most $e^{-2ay}$
by~\eqref{eq:overshootidentity}. Conditioning at time $n$ and changing
Gaussian drift give
\begin{align}
0\le Q_n^a(x)-Q_\infty^a(x)
&\le\E_x^a[e^{-2aX_n};\tau>n]
 =e^{-2ax}Q_n^{-a}(x)\notag\\
&\le e^{-ax-a^2n/2}Q_n^0(x)
 \le C(x+1)n^{-1/2}e^{-a^2n/2}.
\label{eq:finiteinfinite}
\end{align}
The middle inequality uses $W_n\ge-x$ on zero-drift survival.

For the duration, fix $n\ge4E+2$ and put $N=n-1$.
Feasibility through $n$ requires the last lower slope
$m_0=(w+W_N-2E)/n$ to lie below every upper slope.
Conditional on $W_N=y$, set $x=w-m_0$. The carry requires $x\ge0$,
and the other necessary inequalities become
\[
 x+\beta_j+\frac{2E-x}{N}j\ge0,\qquad 0\le j\le N,
\]
where $\beta_j=W_j-(j/N)W_N$ is a Gaussian bridge independent of $y$.
Its first $M=\lfloor N/2\rfloor$ steps have density relative to a
Gaussian walk bounded by
$p_{N-M}(-W_M)/p_N(0)\le\sqrt2$, where $p_t$ is the $N(0,t)$ density.
The displayed drift is at most $2E/N\le1/2$; if negative, replace it
by zero to enlarge survival. Lemma~\ref{lem:ballot} bounds the conditional
probability by $C(x+1)(E/n+n^{-1/2})\ind_{\{x\ge0\}}$.
Since $x_+\le w_++2E/n+|y|/n$, averaging gives
\begin{equation}\label{eq:durationtail}
 \Prob(L_E(w)\ge n)\le C(1+w_+)E/n,
 \qquad4E+2\le n\le n_0:=\lceil E^2\rceil.
\end{equation}

Apply Lemma~\ref{lem:threepointblock} at the accepted-prefix index $n_0$.
Its blocks start at $W_{n_0-1}$ and use only subsequent increments,
so averaging~\eqref{eq:threepointconditional} and then using
\eqref{eq:durationtail} gives
\[
 \Prob(L_E(w)\ge n_0+2r_Em)\le C(1+w_+)E^{-1}p_*^m,
 \qquad
 \sum_{n\ge n_0}\Prob(L_E(w)\ge n)\le CE(1+w_+).
\]
Together with~\eqref{eq:durationtail} and the trivial bound for
$n<4E+2$, this proves~\eqref{eq:meanuniform}.

\subsection{The exact logarithmic coefficient}

Fix $0<\delta<1/2$, and restrict to
$E^{1+\delta}\le n\le E^{2-\delta}$.
For fixed $0<r_0<1$, let
$D_{E,n}=\{\max_{0\le j<n}|w+W_j|\le r_0E\}$.
Uniformly for $w$ in a fixed compact interval, the reflection principle
gives $\Prob(D_{E,n}^c)\le4e^{-c_{r_0}E^\delta}$.
Put $a_\pm=(2\pm r_0)E/n$.
On $D_{E,n}$, infinite lower-band survival with slope $-a_-$ also
satisfies the upper-band constraint through $n$. Conversely, feasibility
and $D_{E,n}$ force the last lower slope to be at least $-a_+$.
For $w\ge0$ this gives
\[
 Q_\infty^{a_-}(w+a_-)-\Prob(D_{E,n}^c)
 \le\Prob(L_E(w)\ge n)
 \le Q_{n-1}^{a_+}(w+a_+)+\Prob(D_{E,n}^c).
\]
Here $a_\pm\to0$ uniformly and $a_\pm^2(n-1)\ge cE^\delta$.
Equations~\eqref{eq:smalldrift} and~\eqref{eq:finiteinfinite}, followed
by $r_0\downarrow0$, imply
\[
 \sup_{E^{1+\delta}\le n\le E^{2-\delta}}
 \left|\frac nE\Prob(L_E(w)\ge n)-4V(w)\right|\longrightarrow0.
\]
All errors remain negligible after division by $E/n$.
Summing over this interval gives
$(4V(w)+o(1))(1-2\delta)E\log E$.
The two omitted logarithmic ranges contribute at most
$C(1+w_+)\delta E\log E+O(E(1+w_+))$
by~\eqref{eq:durationtail}; the tail beyond $E^2$ contributes $O(E(1+w_+))$.
Let $E\to\infty$ and then $\delta\downarrow0$.
This proves the main limit, uniformly on nonnegative compact intervals.

Janssen and van Leeuwaarden~\cite{jvl}, Theorem 1 and equation (2.1),
give $Q_\infty^a(0)\sim\sqrt2 a$ by reflecting their negative-drift
Gaussian walk. The proof of~\eqref{eq:smalldrift} with $x=0$ identifies
the coefficient as $2V(0)$, giving $V(0)=1/\sqrt2$.
Their result supplies the Gaussian overshoot constant; the conditional
Swing duration asymptotic is proved above.

\subsection{The sampled anchored lifetime}
\label{app:firstmeanproof}

\begin{proof}[Proof of Theorem~\ref{thm:firstmean}]
On a single Brownian probability space, $E B(j/E^2)$ has the law of
unit Gaussian partial sums.  Dividing values by $E$ and sample indices
by $E^2$ preserves feasibility and turns the first recording duration
into that of unit-tolerance Brownian sampling on mesh $E^{-2}$.
Theorem~\ref{thm:fixed} therefore gives
$L_{1,E}/E^2\Rightarrow\sigma$.
This coupling proves a distributional assertion for the original walk;
it does not assert almost-sure convergence on one fixed Gaussian sequence.

Condition on its first observation and apply Lemma~\ref{lem:threepointblock}.
For $E\ge1$, its block length satisfies $2r_E\le18E^2$, so
\[
 \Prob(L_{1,E}/E^2>x)\le p_*^{\lfloor x/18\rfloor},\qquad x\ge0.
\]
This integrable envelope gives uniform integrability and hence
$\E L_{1,E}/E^2\to\E\sigma$.
The value of $\E\sigma$ is the anchored Brownian constant in
\cite[Theorem~1.2]{chenBrownian}.
\end{proof}

\subsection{Averaging limits and the boundary at zero}

Averaging a conditional coefficient requires control of the
restart-state distribution.
For example, if actual states $w_E$ converge in law to $W$, have uniformly
integrable positive parts, and $\Prob(W=0)=0$, then
\[
 \frac{\E L_E(w_E)}{E\log E}
 \longrightarrow4\E[V(W)\ind_{\{W>0\}}],
\]
provided the future conditional on each state has the specified fresh
Gaussian law. Use compact uniformity for positive states, the negative
carry tail bound away from zero,~\eqref{eq:meanuniform}, and truncation.
Appendix~\ref{app:transientconvergence} uses this mean-passage argument
for every subsequential law of the actual initialization.
An atom at zero cannot be treated by weak convergence alone: $w_E=0$
gives coefficient $2\sqrt2$, whereas $w_E=-1/\sqrt{\log E}\to0$ gives
coefficient zero. For the latter, the negative carry bound implies
$g_E(-c)\le2+4E/c+8/c^2$, and substitute $c=1/\sqrt{\log E}$.

\section{The fixed-tolerance restart chain}\label{app:stationary}

\begin{proof}[Proof of Theorem~\ref{thm:stationary}]
Lemma~\ref{lem:threepointblock} gives $N<\infty$ almost surely and
$1\le h_E(x)\le H_E$ uniformly over $x$.
Let $Z$ denote a standard normal variable.

Next, comparison of the first two slope intervals gives
\[
 N=1\quad\Longleftrightarrow\quad |W_1-x|>3E.
\]
On this event least squares gives $m=x$ and $x'=W_1$. Hence, with
$\varphi$ denoting the standard normal density,
\[
 P_E(x,dy)\ge\varphi(y)\ind_{\{|y-x|>3E\}}\,dy.
\]
Choose
\[
 A_+=[4E+1,4E+2],\quad A_-=-A_+,\quad B=[8E+4,8E+5].
\]
For $x\le0$, an increment in $A_+$ produces a one-interval segment;
for $x>0$, use $A_-$. From either resulting interval, another fresh
increment in $B$ again produces a one-interval segment.
Let $p_E=\Prob(Z\in A_+)=\Prob(Z\in A_-)>0$,
$q_E=\Prob(Z\in B)>0$, $\alpha_E=p_Eq_E$, and
$\nu_E=\mathcal L(Z\mid Z\in B)$. Then, for every $x$,
\[
 P_E^2(x,\cdot)\ge\alpha_E\nu_E(\cdot).
\]
Writing $Q=P_E^2=\alpha_E\nu_E+(1-\alpha_E)R$ for a probability
kernel $R$ shows that $\mu\mapsto\mu Q$ contracts total variation
by $1-\alpha_E$. The space of probability measures is complete in
this metric, so the contraction has a unique fixed point $\pi_E$.
Since $\pi_EP_E$ is also invariant under $Q$, uniqueness gives
$\pi_EP_E=\pi_E$. Any $P_E$-invariant law is $Q$-invariant, proving
uniqueness for $P_E$ as well. Iterating the contraction and using
nonexpansion for the remaining odd step proves the displayed rate.
Finally, apply it to the nonnegative bounded function $h_E$ in
\eqref{eq:meanrecursion}.
\end{proof}

\section{Uniform stationary margin bounds}\label{app:margins}

\begin{proof}[Proof of Theorem~\ref{thm:stationaryscale}]
All probabilities and expectations in this proof start the exact
state chain in $\pi_E$, unless a starting state or an accepted history
is specified.  Given $x$, let $W_n=Z_1+\cdots+Z_n$ be the fresh
Gaussian walk, and retain the notation $D_i=x+W_{i-1}$, $N$, $m$, and
$x'=x+W_N-mN$ from Section~\ref{sec:stationary}.  Thus $N$ is a stopping
time for the fresh walk and $\{N\ge n\}$ is measurable before $Z_n$.
We first use only the already proved bound $1\le C_E\le C(1+E^2)$.
Constants below are independent of $E$.

\paragraph{A pathwise bound for the completed SSE slope.}
Set
\[
 c_x=(|x|-E)_+,\qquad
 G=\sup_{i\ge1}\frac{|W_{i-1}|}{i},\qquad B=c_x+2G.
\]
The variable $G$ is independent of $x$.  For $t\ge1$,
\[
 \Prob(G>t)
 \le\sum_{i\ge2}2\exp\!\left(-\frac{t^2i^2}{2(i-1)}\right)
 \le C e^{-t^2},
\]
so $G$ has moments of every fixed order.
For $x\ge0$, the candidate slope $c_x+G$ satisfies
\[
 -E\le x-Bi\le D_i-(c_x+G)i\le x-c_xi\le E
 \quad\text{whenever }i\le E/B.
\]
For $x<0$, reflect the data.  Consequently
\begin{equation}\label{eq:stationarydurationlower}
 N\ge\left\lfloor\frac{E}{c_x+2G}\right\rfloor.
\end{equation}
The denominator is positive almost surely.  The unconstrained SSE slope
satisfies
\[
 |m_{\rm LS}|\le\frac{3|x|}{2N+1}+G.
\]
If $B\le E/2$, then $N\ge E/(2B)$ and $|x|\le E+B$, giving
$|m_{\rm LS}|\le5B$.  If $B>E/2$, use $N\ge1$ to obtain
$|m_{\rm LS}|\le E+c_x+G\le3B$.
The cone's lower endpoint is at most $c_x+G$, and its upper endpoint
is at least $-c_x-G$.  Projection of $m_{\rm LS}$ onto this cone
therefore gives
\begin{equation}\label{eq:stationaryslope}
 |m|\le5c_x+10G.
\end{equation}

\paragraph{Two controls on the next margin.}
Let $R'=E-|x'|$.  Since the rejected point cannot fit the old chosen
line, $|x'-m|>E$, and therefore
\begin{equation}\label{eq:stationarymarginpositive}
 (R')_+\le |m|\le5R_-+10G.
\end{equation}
The last accepted residual $r_N=D_N-mN$ lies in $[-E,E]$, while
$x'=r_N+Z_N$.  Hence $(R')_-\le|Z_N|$.  Although $Z_N$ is selected
by the stopping rule, a predictable stopped sum gives
\begin{equation}\label{eq:stationaryrawtail}
 \begin{split}
 \Prob_{\pi_E}(R_->t)
 &=\Prob_{\pi_E}((R')_->t)\\
 &\le\E_{\pi_E}\sum_{n=1}^N\ind_{\{|Z_n|>t\}}
 =C_E\Prob(|Z_1|>t)\le2C_Ee^{-t^2/2}.
 \end{split}
\end{equation}
Combining this with~\eqref{eq:stationarymarginpositive} and the tail
of $G$ yields absolute $c_0,C_0>0$ such that
\begin{equation}\label{eq:stationaryrawmargintail}
 \Prob_{\pi_E}(|R|>t)\le C_0(C_E+1)e^{-c_0t^2},\qquad t\ge20.
\end{equation}
In particular, the fixed-$E$ state moments used below are finite.

\paragraph{A predictable favorable rejection branch.}
Let $\mathcal G$ be the union of two completion events: the old SSE
slope is $m=a_N\in[0,E]$ and the new point is a lower rejection; or
$m=b_N\in[-E,0]$ and it is an upper rejection.  Here $[a_N,b_N]$
is the old nonempty cone.  For each prospective final index $n$, its
accepted history, cone, and SSE slope $m_n$ are measurable before
$Z_n$.  This predictability will be used when conditioning on the
last increment.

We first estimate the probability of the complementary branch.
Take $\rho=1/40$ and $\delta=1/20$.  On the sufficient event
\begin{equation}\label{eq:stationarypathgood}
 |R|\le\rho E,\qquad \max_{0\le j\le N}|W_j|\le\rho E,
\end{equation}
all $D_i$, including the rejected $D_{N+1}$, have the sign of $x$
and absolute values in $[(1-\delta)E,(1+\delta)E]$.
Suppose first that they are positive.  The old upper slope bound is
at least $(2-\delta)E/N$, whereas the incoming lower slope is at
most $\delta E/(N+1)$.  Upper rejection is impossible.  Lower
rejection instead implies
\[
 a_N>\frac{(2-\delta)E}{N+1},\qquad a_N\le\delta E,
\]
and thus $N\ge39$.  Moreover,
\[
 m_{\rm LS}\le\frac{3(1+\delta)E}{2N+1}
 \le\frac{(2-\delta)E}{N+1}<a_N.
\]
The middle inequality is equivalent to
$(1-5\delta)N\ge1+4\delta$.  Thus SSE clips to
$m=a_N\in(0,\delta E]$.  Negative data give the reflected branch.
This proves that~\eqref{eq:stationarypathgood} implies $\mathcal G$.

Wald's second-moment identity and the maximal inequality for the
stopped square submartingale give
\[
 \E_{\pi_E}W_N^2=C_E,\qquad
 \Prob_{\pi_E}\!\left(\max_{j\le N}|W_j|>\rho E\right)
 \le\frac{C_E}{\rho^2E^2}.
\]
Together with~\eqref{eq:stationaryrawmargintail} and the coarse
bound on $C_E$, this gives, for all sufficiently large $E$,
\begin{equation}\label{eq:stationarybadbranch}
 \Prob_{\pi_E}(\mathcal G^c)\le K\frac{C_E}{E^2}.
\end{equation}

\paragraph{A uniform tail on the favorable branch.}
Consider an old accepted history with $m_n=a_n\in[0,E]$, and write
its last residual as $r_n=D_n-m_nn=-E+g$, with $g\ge0$.
The next increment $\zeta=Z_n$ causes a lower rejection exactly when
\begin{equation}\label{eq:stationaryrejectthreshold}
 \zeta<m_n-g.
\end{equation}
On this event $x'=r_n+\zeta<m_n-E\le0$, so
$(R')_-=(-g-\zeta)_+$.  Consequently, for $t\ge0$,
\[
 \frac{\Prob(\zeta<-g-t)}{\Prob(\zeta<m_n-g)}
 \le\frac{\Phi(-g-t)}{\Phi(-g)}\le e^{-t^2/2}.
\]
The upper rejection branch is its reflection.  Summing over the
predictable accepted histories and the disjoint final-index and
rejection-direction events gives
\[
 \Prob_{\pi_E}((R')_->t,\mathcal G)\le e^{-t^2/2}.
\]
This conditions only on the accepted history and the indicated
rejection direction, not on the future-dependent sufficient event
\eqref{eq:stationarypathgood}.  On the complementary branch,
\eqref{eq:stationarybadbranch} and the stopped-sum bound give
\begin{equation}\label{eq:stationaryrefinedtail}
 \Prob_{\pi_E}(R_->t)
 \le e^{-t^2/2}
 +\min\!\left\{K\frac{C_E}{E^2},\,2C_Ee^{-t^2/2}\right\}.
\end{equation}
Integrate at $t=\sqrt{8\log E}$, using the two alternatives on
the respective sides of this threshold.  For every fixed $p>0$,
\begin{equation}\label{eq:stationarynegativemoment}
 \E_{\pi_E}R_-^p
 \le K_p+K_p\frac{C_E}{E^2}(\log E)^{p/2}.
\end{equation}
The Gaussian integral above the threshold has an $E^{-4}$ factor,
which absorbs every fixed power of $\sqrt{\log E}$.

\paragraph{Absorption and stationary moment bounds.}
Take $p=1$ in~\eqref{eq:stationarynegativemoment}, and then use
\eqref{eq:stationarymarginpositive} and $\E G<\infty$.
The uniform conditional mean bound~\eqref{eq:meanuniform} gives
\begin{equation}\label{eq:stationaryabsorption}
 C_E\le KE\log E
       +K\frac{(\log E)^{3/2}}{E}\,C_E.
\end{equation}
The last coefficient tends to zero, so it can be absorbed into the
left side.  This proves $C_E\le KE\log E$.
Substitution into~\eqref{eq:stationarynegativemoment} now gives
\[
 \E_{\pi_E}R_-^p
 \le K_p+K_p\frac{(\log E)^{1+p/2}}{E}\le K'_p.
\]
Equation~\eqref{eq:stationarymarginpositive} and the moments of $G$
give the same conclusion for $R_+$, proving
\eqref{eq:stationarymarginmoments}.  Equation~\eqref{eq:stationaryslope}
also bounds every fixed stationary moment of $|m|$ uniformly.
In addition, $\Prob_{\pi_E}(\mathcal G^c)=O(\log E/E)$.

\end{proof}

\section{Identification of the limiting restart law}\label{app:limit}

\begin{proof}[Proof of Theorem~\ref{thm:leadingcoefficient}]
We first identify the distant-crossing representation and then pass to
the invariant laws.

\paragraph{The exact distant-crossing representation.}
By reflection, a margin $r$ is represented by input state $x=E-r$.
The observations relative to the current fitted anchor are
$D_i=E-r+W_{i-1}$, and the lower slope bound is
\[
 a_n=\max_{i\le n}\frac{W_{i-1}-r}{i}.
\]
Along any fixed path, for all sufficiently large $E$ the finite maximizing
index in~\eqref{eq:leadingslope} is in the accepted prefix, so $a_N=A_r=:a$.
The frozen-state duration limit, after reflection, gives $N/E\to2/a$.
The strong law then yields
\[
 \frac{m_{\rm LS}}a
 =\frac{(E-r)\sum_{i\le N}i+\sum_{i\le N}iW_{i-1}}
        {a\sum_{i\le N}i^2}\longrightarrow\frac34.
\]
Hence constrained least squares selects $m=a$.  The first violated upper
constraint gives
\begin{equation}\label{eq:distantcrossing}
 N=\inf\{n\ge0:an-W_n>2E-r-a\}.
\end{equation}
Writing $B=2E-r-a$ and $O_B=aN-W_N-B>0$, the exact next state is
\[
 x'=E-r+W_N-aN=-E+a-O_B.
\]
For $E>a$ it is negative, and therefore
\begin{equation}\label{eq:marginovershootidentity}
 R'=E-|x'|=a-O_B.
\end{equation}
The rejected observation is present in this overshoot.  Replacing it by an
independent observation would change the transition kernel.

\paragraph{A uniform renewal statement.}
We first verify the renewal estimate needed when the drift varies over a
compact interval $[\delta,M]\subset(0,\infty)$.  Conditional on the position
$-d\le0$ just before a first upcrossing, the final overshoot is distributed as
$a+Z-d$ conditional on being positive.  Monotonicity of the Gaussian hazard
function gives, after mixing over the disjoint crossing histories,
\begin{equation}\label{eq:ladderuniformtail}
 \Prob(H_a>t)\le
 \frac{\overline\Phi(t-M)}{\overline\Phi(-M)},\qquad t\ge0.
\end{equation}
Consequently every fixed exponential moment is uniformly bounded.
The event of crossing on the first step also gives
\[
 \operatorname{Law}(H_a)(dh)\ge
 \varphi(h-a)\ind_{\{h>0\}}\,dh.
\]
On $[1,2]$ this contains a common positive multiple of the uniform law.
It follows that the ladder-height laws are uniformly strongly nonlattice and
that $\inf_{a\in[\delta,M]}\mu_a>0$, where $\mu_a=\E H_a$.

Blanchet and Glynn~\cite[Theorem~1(i), p.~1073]{bg} give an exponential
renewal remainder for a uniformly strongly nonlattice family with
uniformly bounded exponential moments, with a factor $\mu_a^2$ on the
remainder.  The uniform positive lower bound for $\mu_a$ verified above
therefore gives, for these ladder heights,
\begin{equation}\label{eq:uniformladderrenewal}
 U_a(u)=\frac{u}{\mu_a}+
 \frac{\E H_a^2}{2\mu_a^2}+D_a(u),\qquad
 |D_a(u)|\le Ce^{-cu},\qquad u\ge0,
\end{equation}
with $c,C>0$ independent of $a\in[\delta,M]$.
The overshoot $O_a(b)$ of the original positive-drift walk above $b$ is the
residual life of its strict ladder-height renewal process.  Thus
\begin{equation}\label{eq:overshootrenewalidentity}
 \Prob(O_a(b)>t)=\int_{[0,b]}\Prob(H_a>b+t-s)\,U_a(ds).
\end{equation}
Stieltjes integration by parts in~\eqref{eq:overshootrenewalidentity}, using
\eqref{eq:uniformladderrenewal}, shows that the linear term converges to
\[
 \mu_a^{-1}\int_t^\infty\Prob(H_a>v)\,dv.
\]
This is the tail of~\eqref{eq:equilibriumovershoot}.  The constant term leaves only a distant
ladder-height tail.  For the remainder, split its Stieltjes integral at
$s=b/2$: the first part is bounded by the distant tail
\eqref{eq:ladderuniformtail}, and the second by $Ce^{-cb/2}$.
The resulting tail convergence is uniform in $a\in[\delta,M]$.
Truncating the overshoot variable then gives the same uniform convergence
for bounded Lipschitz test functions.  No renewal estimate uniform as
$a\downarrow0$ is required below.

\paragraph{Decoupling the random slope from the distant overshoot.}
Conditioning directly on $A_r$ would condition on the entire future.
Instead fix $K_0<\infty$, let
\[
 A_{r,J}=\max_{1\le i\le J}\frac{W_{i-1}-r}{i},
\]
and consider all $|r|\le K_0$.  Given any positive error probability, first
choose $\delta>0$, and then $M,J<\infty$, so that the event
\begin{equation}\label{eq:earlyslopelocalization}
 2\delta\le A_{K_0,J}\le A_{-K_0,J}\le M,
 \qquad
 \sup_{i>J}\frac{|W_{i-1}|+K_0}{i}<\delta
\end{equation}
has complement smaller than that error.  This is possible because
$A_{K_0}>0$, $A_{-K_0}<\infty$, and the strong law holds.
On this event, $A_r=A_{r,J}\in[2\delta,M]$ simultaneously for every
$|r|\le K_0$.

For a substitute output, use only the prefix through $W_{J-1}$ to set
$a=A_{r,J}$, and let the independent increments after that prefix cross the
level
\begin{equation}\label{eq:prefixcrossinglevel}
 b_{E,r,J}=2E-r-aJ+W_{J-1}.
\end{equation}
Define the substitute as zero if $a\le0$ or $b_{E,r,J}\le0$; otherwise it is
$a$ minus this overshoot.  On~\eqref{eq:earlyslopelocalization} those exceptional
cases eventually disappear.  For large $E$, slope $a$ satisfies the finitely
many constraints through $J$.  All subsequent lower bounds are at most $a$,
so failure occurs exactly when an upper bound falls below $a$.
The first crossing is in $[E/M,3E/\delta]$ uniformly over $|r|\le K_0$.
The strong law on this range gives uniformly $N/E=2/a+o(1)$ and
$m_{\rm LS}/a=3/4+o(1)$.  Hence the true and substitute outputs eventually
coincide simultaneously on~\eqref{eq:earlyslopelocalization}.

This event is used only to bound the replacement error; it is not included
in a conditional law of the future.  In the actual conditional expectation
we condition only on $\sigma(W_0,\ldots,W_{J-1})$.  Then $a$ and
$b_{E,r,J}$ are known, and the remaining increments are independent.
The uniform renewal result applies for $a\in[2\delta,M]$; the crossing level
tends to infinity uniformly in $|r|\le K_0$ for each fixed prefix.
First let $E\to\infty$ and then send the truncation error to zero.
For every bounded Lipschitz $f$, the original reflected-margin kernel $K_E$
therefore satisfies
\begin{equation}\label{eq:marginkernellocalconvergence}
 \sup_{|r|\le K_0}|K_Ef(r)-Kf(r)|\longrightarrow0.
\end{equation}
For $E\ge K_0$ this compact set lies in the domain $r\le E$ of $K_E$.

\paragraph{Continuity and uniqueness of an invariant law.}
The map $r\mapsto A_r$ is pathwise 1-Lipschitz.
To examine the second part of the kernel, couple all drifts using the same
Gaussian walk.  For each fixed $a>0$, the first crossing time is finite and
all positions with indices 1 through that time are nonzero almost surely.  If $a_j\to a$,
the crossing times are eventually identical, and $H_{a_j}\to H_a$ almost surely.
The uniform tail bound~\eqref{eq:ladderuniformtail} gives
$\E|H_{a_j}-H_a|\to0$ and $\mu_{a_j}\to\mu_a>0$.  Moreover,
\[
 \int_0^\infty|\Prob(H_{a_j}>o)-\Prob(H_a>o)|\,do
 \le\E|H_{a_j}-H_a|\longrightarrow0.
\]
Thus $q_{a_j}\to q_a$ in $L^1$.  Translation continuity in $L^1$ implies
that $\operatorname{Law}(a_j-O_{a_j})$ converges in total variation to
$\operatorname{Law}(a-O_a)$.  Dominated convergence and the pathwise
continuity of $A_r$ show that $Kf$ is continuous for every bounded measurable
$f$: the kernel is strong Feller.

For each $r,y\in\R$, the density in~\eqref{eq:limitingmargindensity} is strictly
positive.  Indeed $A_r$ has unbounded upper support already from its term
$(W_1-r)/2$, and $H_a$ has unbounded upper support from a first-step crossing.
Consequently every $K(r,\cdot)$ is equivalent to Lebesgue measure, as is any
invariant probability.
For a direct uniqueness argument, suppose that $\nu_1,\nu_2$ are invariant,
put $\nu=(\nu_1+\nu_2)/2$ and $f=d\nu_1/d\nu\in[0,2]$, and take a stationary
pair $X\sim\nu$, $Y\mid X\sim K(X,\cdot)$.  Invariance gives
$f(Y)=\E[f(X)\mid Y]$.  Since $X$ and $Y$ have the same law, equality holds
in the conditional quadratic Jensen inequality, forcing $f(X)=f(Y)$ almost
surely.  The pair law is equivalent to $\nu\otimes\nu$, so $f$ is constant.
Normalization gives $\nu_1=\nu_2$.

\paragraph{Passing the stationary laws and the means to the limit.}
Theorem~\ref{thm:stationaryscale} gives tightness of $\nu_E$ and uniform
integrability of every fixed positive power of the margins.
If $\nu_{E_j}\Rightarrow\nu$, then
\eqref{eq:marginkernellocalconvergence}, tightness, continuity of $Kf$, and
$\nu_EK_E=\nu_E$ imply $\nu Kf=\nu f$ for every bounded Lipschitz $f$.
This class determines measures, so $\nu$ is invariant.
The uniqueness just proved implies existence of a unique limit $\nu_\infty$
and~\eqref{eq:marginweaklimit}.  Its density implies in particular
$\nu_\infty(\{0\})=0$.

On each compact positive interval,
$g_E(r)/(E\log E)\to4V(r)$ uniformly by
Theorem~\ref{thm:conditionalmean}.  On compact negative intervals bounded away
from zero it converges uniformly to zero, by the Gaussian tail estimate in
the negative-state part of Theorem~\ref{thm:frozen}.
The uniform bound~\eqref{eq:meanuniform} controls a neighborhood of zero by
its probability mass; that mass vanishes as the neighborhood shrinks,
because the limit has no atom there.  The same bound and the uniform margin
moments control the remote tails.  Integrating the conditional means against
$\nu_E$ proves~\eqref{eq:leadingcoefficient}.  Its value is finite since
$V(r)\le r+C$, and positive since $V(r)>0$ on the positive half-line and
$\nu_\infty$ is equivalent to Lebesgue measure.
\end{proof}

\section{Uniform moments at the first restart}
\label{app:firstrestart}

The first restart does not require a separate microscopic entrance law
for the moment bounds used below.  We control the actual selected
rejection increment, keeping its conditioning on first rejection.

\begin{lemma}[First-restart moments]\label{lem:firstmarginmoments}
Let $Y_0=0$, let $Y_n=\sum_{i=1}^nZ_i$ with iid standard Gaussian
increments, and initialize Swing at its true first sample.  Write
$N=L_{1,E}$, $J=N+1$, let $[a_N,b_N]$ be the last feasible slope
interval, and let $m$ be its selected slope.  Put
$R_{1,E}=E-|Y_J-Nm|$.  There are absolute constants $c,C>0$ such that,
for all $E\ge8$ and $t\ge0$,
\begin{equation}\label{eq:firstselectedtail}
 \Prob(|Z_J|>t)\le Ce^{-ct}.
\end{equation}
Consequently, for every fixed $p>0$,
\begin{equation}\label{eq:firstmarginmoments}
 \sup_{E\ge8}\E|R_{1,E}|^p<\infty.
\end{equation}
\end{lemma}

\begin{proof}
We first bound both endpoints of the final feasible interval.  In the
state coordinates of \eqref{eq:stateupdate}, write $x=Z_1$ and
$W_n=Z_2+\cdots+Z_{n+1}$, so $Y_i=x+W_{i-1}$.  Set
\[
 c_x=(|x|-E)_+,\qquad
 G=\sup_{i\ge1}\frac{|W_{i-1}|}{i},\qquad B=c_x+2G.
\]
The candidate-line bound \eqref{eq:stationarydurationlower} gives
$N\ge\lfloor E/B\rfloor$, and always $N\ge1$.  The cone inequalities
give $a_N\le c_x+G$, $b_N\ge-c_x-G$; using its final accepted sample
also gives
\[
 a_N\ge-\frac{|x|+E}{N}-G,
 \qquad b_N\le\frac{|x|+E}{N}+G.
\]
If $B\le E/2$, then $N\ge E/(2B)$ and $|x|\le E+B$, hence
$(|x|+E)/N+G\le6B$.  If $B>E/2$, the same bound follows from
$N\ge1$ and $2E+c_x+G\le5B$.  Therefore
\begin{equation}\label{eq:firstconebound}
 \max(|a_N|,|b_N|)\le6B=6c_x+12G.
\end{equation}
A Gaussian union bound gives $\Prob(G>t)\le Ce^{-ct^2}$ for $t\ge1$.
Since $c_x\le|Z_1|$, it follows that
\begin{equation}\label{eq:firstconetail}
 \Prob\bigl(\max(|a_N|,|b_N|)>t\bigr)\le Ce^{-ct^2},
 \qquad E\ge8,\quad t\ge0.
\end{equation}

We next compare one-step and two-step first-rejection events.  On
$\{J>n\}$, condition on $\mathcal F_n$, and denote the old cone by
$[\ell,u]$.  An upper rejection at the next sample occurs exactly when
\[
 Z_{n+1}>r:=(n+1)u+E-Y_n.
\]
Because slope $u$ fits the last accepted sample, $r\ge u$.
Fix large $t$ and a cutoff $M>0$.  The contribution of an upper rejection
with $u>M$ is bounded by \eqref{eq:firstconetail}, since this $u$ is
then the final endpoint $b_N$.

On prefixes with $u\le M$ and $r\le t/2$, a Gaussian tail ratio gives
\[
 \bar\Phi(t)\le Ce^{-ct^2}\bar\Phi(r).
\]
Thus, after summing over the disjoint first-rejection events, these
prefixes contribute at most $Ce^{-ct^2}$ to the probability of upper
rejection with $Z_J>t$.

For the remaining prefixes, $u\le M$ and $r>t/2>1$, replace the
prospective increment by a value $z=r-d$ with $1/2\le d\le1$.
The two new slope constraints are then
\[
 u-\frac{d}{n+1},\qquad u+\frac{2E-d}{n+1}.
\]
They intersect the old cone and leave its upper endpoint equal to $u$.
The sample is therefore accepted.  At the following sample, the upper
rejection threshold is exactly $u+d\le M+1$.  Define
\[
 p(r)=\int_{r-1}^{r-1/2}\varphi(z)\,dz.
\]
Given this prefix, the probability of the described acceptance followed
by an upper first rejection is at least $p(r)\bar\Phi(M+1)$.
For $r>1$, monotonicity of the Gaussian density and Mills' bound give
\[
 \frac{\bar\Phi(r)}{p(r)}
 \le\frac{2\varphi(r)}{r\varphi(r-1/2)}
 \le Ce^{-r/2}.
\]
The original probability of upper rejection with $Z_{n+1}>t$ is at
most $\bar\Phi(r)$, and hence at most $Ce^{-t/4}p(r)$.

To sum this comparison correctly, let $B_n$ be the proxy event that
$J>n$, the prefix satisfies $u\le M,r>t/2$, the next increment lies
in $[r-1,r-1/2]$, and the first rejection occurs from above at $n+2$.
The events $B_n$ are disjoint as $n$ varies because each specifies the
first rejection index.  Their total probability is at most one.  Taking
expectations of the preceding conditional comparisons and summing thus
gives
\begin{equation}\label{eq:firstchargingtail}
 \Prob(Z_J>t,\text{upper rejection})
 \le Ce^{-cM^2}+Ce^{-ct^2}
       +\frac{Ce^{-t/4}}{\bar\Phi(M+1)}.
\end{equation}
There is no factor $\E J$ in the last term, and the selected increment
$Z_J$ has not been treated as an independent Gaussian variable.

Take $M=\sqrt t/4$.  Mills' lower bound implies, for all sufficiently
large $t$,
\[
 \frac{e^{-t/4}}{\bar\Phi(M+1)}
 \le C(1+M)\exp\{-t/4+(M+1)^2/2\}\le Ce^{-ct}.
\]
Equation~\eqref{eq:firstchargingtail} now has an exponential upper
bound.  An upper rejection with $Z_J<-t$ forces $b_N<-t$, since
$Z_J>r\ge b_N$; its probability is covered by
\eqref{eq:firstconetail}.  Reflecting the path handles lower rejection.
Increasing the constant to cover bounded $t$ proves
\eqref{eq:firstselectedtail}.

Finally, the last accepted residual $Y_N-Nm$ lies in $[-E,E]$, so
$(R_{1,E})_-\le|Z_J|$.  Rejection of the selected old line gives
$|Y_J-(N+1)m|>E$, and hence $(R_{1,E})_+\le|m|$.
Since $m\in[a_N,b_N]$, \eqref{eq:firstconetail} and
\eqref{eq:firstselectedtail} imply \eqref{eq:firstmarginmoments} by
tail integration.  The constants may depend on $p$, but not on $E$.
\end{proof}

In the relative-state notation, the final fresh increment is $W_N-W_{N-1}$;
it is the original increment $Z_{N+1}=Z_J$.  This one-index shift is why
the proof above uses the actual first-rejection index throughout.

\section{The special second segment}
\label{app:secondsegment}

We use unit Gaussian increments, a true initial anchor, the original
constrained least-squares selector, the carried rejected sample, and no
maximum-lag constraint. Constants below are absolute unless indicated
otherwise.

\begin{lemma}\label{lem:secondnegativemean}
For all sufficiently large $E$ and every $c>0$,
\begin{equation}\label{eq:secondnegativemean}
 g_E(-c)\le CE[1+\log_+(1/c)].
\end{equation}
\end{lemma}
\begin{proof}
We retain the necessary endpoint event in the Gaussian bridge estimate
used to prove Theorem~\ref{thm:conditionalmean}. Put $N=n-1$ and
$w=-c$. Conditional on $W_N=y$, its starting allowance is
$x=(2E-cN-y)/n$. For $4E+2\le n\le n_0=\lceil E^2\rceil$, that
estimate gives
\[
 \Prob(L_E(-c)\ge n)
 \le \frac{CE}{n}\E[(1+x)\ind_{\{x\ge0\}}].
\]
If $N\ge4E/c$, put $b=2E-cN\le-cN/2$. Gaussian integration yields
\[
 \E[(1+x)\ind_{\{x\ge0\}}]
 =\Phi(b/\sqrt N)
   +\frac{b\Phi(b/\sqrt N)+\sqrt N\varphi(b/\sqrt N)}n
 \le C e^{-c^2N/8}.
\]
Consequently
\begin{equation}\label{eq:secondnegativeexponentialtail}
 \Prob(L_E(-c)\ge n)
 \le\frac{CE}{n}e^{-c^2(n-1)/8},
 \qquad n-1\ge4E/c,\quad n\le n_0.
\end{equation}

For $8/E\le c\le1$, put $n_c=\lceil4E/c\rceil+2\le n_0$.
The trivial initial bound and \eqref{eq:durationtail} give a total
of at most $CE[1+\log(1/c)]$ below $n_c$. Above $n_c$, the sum in
\eqref{eq:secondnegativeexponentialtail} is bounded by
\[
 CE+CE\int_{4E/c}^{\infty}\frac{e^{-c^2u/8}}u\,du
 \le CE,
\]
since the lower limit becomes $Ec/2\ge4$ after substitution.
The block tail beyond $n_0$ in the proof of
Theorem~\ref{thm:conditionalmean} contributes at most $CE$.

If $0<c<8/E$, the uniform bound \eqref{eq:meanuniform} is absorbed
by the right side of \eqref{eq:secondnegativemean}. If $c>1$,
the carried sample forces slope at most $-c$, so the last-point
Gaussian tail gives
$g_E(-c)\le2+4E/c+8/c^2\le CE$.
These three ranges prove the claim.
\end{proof}

\begin{lemma}\label{lem:outputanticoncentration}
Fix $B>0$ and $E\ge1$. Consider one completed segment with any given
input history and fresh Gaussian increments after its carried sample.
If its selected feasible slope depends only on the old accepted prefix,
then its output margin $R'=E-|X'|$ satisfies
\begin{equation}\label{eq:outputanticoncentration}
 \Prob(R'\in I\mid\hbox{input history})\le C_B|I|
 \qquad\text{for every interval }I\subset[-B,B].
\end{equation}
The constant does not depend on the input history, $E$, or the segment
index. In particular this holds for the original constrained
least-squares selector.
\end{lemma}
\begin{proof}
Condition throughout on the input history. Write $D_i$ for the data
relative to the old anchor, with the carry $D_1$ already known, and
let $T$ be the relative rejection index. On a surviving accepted prefix
through $n$, the old cone $[\ell,u]$ and selected slope $m$ are
measurable before the next increment $Z$. Here $\F_n$ includes the
input history and the observations through $D_n$. Put
\[
 r=(n+1)u+E-D_n.
\]
An upper rejection occurs when $Z>r$, and then
\[
 X'=D_n+Z-nm>E+u+n(u-m)\ge E+u.
\]
If $u>B$, this implies $R'<-B$, so the target event is impossible.
We may restrict attention to prefixes with $u\le B$.

Let $d=|I|$. The inverse image of $I$ under
$z\mapsto E-|D_n+z-nm|$ consists of at most two intervals of total
length at most $2d$. If $r\le1$, the Gaussian density bound and
$\overline\Phi(r)\ge\overline\Phi(1)$ therefore give
\[
 \Prob(R'\in I,T=n+1,\text{ upper}\mid\F_n)
 \le C_0d\,\Prob(T=n+1,\text{ upper}\mid\F_n),
 \qquad C_0=\frac{2\varphi(0)}{\overline\Phi(1)}.
\]

If $r>1$, the target probability is at most $2d\varphi(r)$.
For comparison, require the next increment to be
$Z=r-a$, with $1/2\le a\le1$. The resulting observation is
$(n+1)u+E-a$. Its lower and upper slope constraints are
\[
 u-\frac a{n+1}
 \quad\text{and}\quad
 u+\frac{2E-a}{n+1}.
\]
Since $2E\ge a$, this observation is accepted and the upper endpoint
remains $u$. The following increment causes an upper rejection with
threshold $u+a\le B+1$, so its conditional probability is at least
$\beta_B=\overline\Phi(B+1)>0$. On the proposed first-increment
interval $[r-1,r-1/2]$, the Gaussian density is at least $\varphi(r)$.
Thus the two-step event has probability at least
$\beta_B\varphi(r)/2$, and
\[
 \Prob(R'\in I,T=n+1,\text{ upper}\mid\F_n)
 \le \frac{4d}{\beta_B}
       \Prob(T=n+2,\text{ upper}\mid\F_n).
\]
The event on the right is a first rejection: its preceding observation
has just been shown to be accepted. No restriction on its eventual
output margin is imposed.

The lower-rejection argument is the reflection of the preceding one.
Average over each surviving prefix and sum over $n\ge1$ and both
directions. The one-step events are disjoint by their first rejection
index; so are the two-step events, indexed by $n+2$. Each total
probability is at most one. Hence \eqref{eq:outputanticoncentration}
holds with $C_B=C_0+4/\beta_B$. This summation introduces no expected
duration factor and makes no Gaussian-law assumption about a selected
increment.
\end{proof}

\begin{proof}[Proof of Theorem~\ref{thm:secondsegment}]
Write $Y_0=0$, $Y_n=\sum_{i=1}^nZ_i$, and let
$N=L_{1,E}$ and $J=N+1$ be the first accepted duration and true
rejection time. Let $m$ be the selected first slope and $q=Nm$ its
recording at time $N$. Define
\begin{equation}\label{eq:secondmargin}
 X_1=Y_J-q,\qquad R_1=E-|X_1|.
\end{equation}

Given $\F_J$, the next segment sees $X_1+W_{i-1}$, $i\ge1$,
with fresh Gaussian increments after the carry. If $X_1\le0$, its
slope constraints are
\[
 \left[\frac{R_1+W_{i-1}-2E}{i},
       \frac{R_1+W_{i-1}}i\right].
\]
If $X_1>0$, reflection gives the same form with $W$ replaced by $-W$.
The reflection is determined by $\F_J$ and preserves the conditional
Gaussian law. Thus, exactly,
\begin{equation}\label{eq:secondconditional}
 \E[L_{2,E}\mid\F_J]=g_E(R_1).
\end{equation}
Because the rejected sample violates every old feasible line,
$|Y_J-(q+m)|>E$. Consequently
\begin{equation}\label{eq:secondmarginpositive}
 (R_1)_+\le |m|.
\end{equation}

The first-cone tail bound~\eqref{eq:firstconetail} gives
$\E|m|^2\le C$ uniformly for large $E$.
Take $n_*=\lceil E^{3/2}\rceil$. If $N<n_*$, slope zero has failed
by time $n_*$; Gaussian embedding and the reflection principle give
\[
 \Prob(N<n_*)\le
 \Prob\left(\max_{i\le n_*}|Y_i|>E\right)
 \le4e^{-c\sqrt E}.
\]
On $\{N\ge n_*\}$ the selected line fits the observation at $n_*$,
so $|m|\le(E+|Y_{n_*}|)/n_*$. Cauchy--Schwarz therefore gives
\[
 \E|m|\le\frac E{n_*}+\frac{\E|Y_{n_*}|}{n_*}
 +(\E|m|^2)^{1/2}\Prob(N<n_*)^{1/2}
 =O(E^{-1/2}).
\]
In particular,
\begin{equation}\label{eq:secondpositiveexpectation}
 \E(R_1)_+=O(E^{-1/2}).
\end{equation}

Lemma~\ref{lem:outputanticoncentration} implies
$\Prob(|R_1|\le\delta)\le C\delta$ for $0<\delta\le1$, and
$\Prob(R_1=0)=0$. By layer-cake integration,
\begin{equation}\label{eq:secondlogarithmicmoment}
 \begin{split}
 \E\left[\log_+\frac1{-R_1};R_1<0\right]
 &=\int_0^\infty\Prob(-e^{-t}<R_1<0)\,dt\\
 &\le\int_0^\infty\min(1,Ce^{-t})\,dt\le C.
 \end{split}
\end{equation}
Equations~\eqref{eq:secondnegativemean} and
\eqref{eq:secondlogarithmicmoment} bound the negative-input
contribution to \eqref{eq:secondconditional} by $CE$.

For the positive contribution, put $\delta=E^{-1/4}$.
The same concentration bound and Markov's inequality give
\[
 \Prob(R_1>0)
 \le C\delta+\delta^{-1}\E(R_1)_+
 \le CE^{-1/4}.
\]
Using \eqref{eq:meanuniform} and
\eqref{eq:secondpositiveexpectation}, we obtain
\[
 \E[g_E(R_1);R_1>0]
 \le CE\log(eE)
       \{\Prob(R_1>0)+\E(R_1)_+\}
 \le CE^{3/4}\log E=o(E).
\]
Together these estimates prove $\E L_{2,E}\le CE$.

For the matching lower bound, Lemma~\ref{lem:firstmarginmoments} lets
us choose a fixed $M\ge1$ so that $\Prob(|R_1|\le M)\ge1/2$ for all
sufficiently large $E$. The independent future envelope
$G=\sup_{n\ge0}|W_n|/(n+1)$ satisfies $\Prob(G\le M)>0$ after
increasing $M$ if necessary. On their intersection, the pathwise
bound~\eqref{eq:stationarydurationlower} gives
\[
 L_{2,E}\ge\left\lfloor\frac{E}{(R_1)_-+2G}\right\rfloor
 \ge\frac{E}{3M}-1.
\]
This event has probability bounded below independently of $E$.
Taking expectations proves $\E L_{2,E}\ge cE$ and completes the theorem.
\end{proof}

\section{An explicit kernel for the stationary coefficient}
\label{sec:newconstant}

We reduce $\gamma_{\rm Swing}$ to a one-dimensional slope kernel,
specify its Gaussian inputs, and give one convergent evaluation.
The same contraction supplies the transient bound in
Theorem~\ref{thm:transientconvergence}. Throughout, the underlying
algorithm uses unit Gaussian increments, constrained least squares,
the carried rejected sample, and no maximum lag.

\subsection{Eliminating the equilibrium overshoot}
For $a>0$, let
\[
 M_a=\sup_{n\ge0}(W_n-an),\qquad
 F_a(x)=\Prob(M_a\le x),\qquad p_a=F_a(0),
\]
with $F_a(x)=0$ for $x<0$, and put
\begin{equation}\label{eq:newconstantG}
 G_a(r)=\E\Phi(r-M_a),\qquad r\in\R.
\end{equation}
The Lindley identity $M_a\overset d=(M_a+Z-a)_+$, with an independent
$Z\sim N(0,1)$ on the right, yields
\begin{equation}\label{eq:newconstantLindley}
 F_a(x)=G_a(x+a),\qquad x\ge0.
\end{equation}
For $r<a$, $G_a(r)$ remains strictly positive and must not be replaced by
$F_a(r-a)$, which is zero under our convention.

Let $H_a$ be the first strict ascending ladder height of a walk with
increments $N(a,1)$, and let $\mu_a=\E H_a$.
Reverse a finite path before its first positive crossing: the event that
all partial sums are nonpositive becomes the event that the terminal
point is a new descending record.  Noninitial ties have probability zero.
Summing over path lengths, including length zero, identifies the killed
occupation measure on the negative half-line with the descending
ladder-height renewal measure.  After reflection this is
$\operatorname{Law}(M_a)/p_a$, by the compound-geometric representation
of a negative-drift maximum.  Consequently
\begin{equation}\label{eq:newconstantheightdensity}
 f_{H_a}(h)=\frac1{p_a}\E\varphi(h+M_a-a),\qquad h>0.
\end{equation}
Set $D=(a-M_a-Z)_+$.  The first-moment Lindley balance gives $\E D=a$;
integrating~\eqref{eq:newconstantheightdensity} against $h$ therefore gives
\begin{equation}\label{eq:newconstantmeanheight}
 \mu_a p_a=a.
\end{equation}
Integrating the same density from $a-r$ to infinity and dividing by
$\mu_a$ eliminates the equilibrium overshoot from the conditional next-margin
law:
\begin{equation}\label{eq:newconstantoutput}
 j_a(r):=\frac{\Prob(R'\in dr\mid A=a)}{dr}
 =\frac1aG_a(r)\ind_{\{r<a\}}.
\end{equation}
Its normalization can also be checked directly:
\[
 \int_{-\infty}^aG_a(r)\,dr
 =\E(a-M_a-Z)_+=a.
\]
Since $\Prob(A_r\le a)=F_a(r+a)$, the limiting margin kernel is
\begin{equation}\label{eq:newconstantkernel}
 K(r,dy)=\left[\int_{(0,\infty)}
 \frac{G_a(y)}a\ind_{\{y<a\}}\,d_aF_a(r+a)\right]dy.
\end{equation}
The Stieltjes measure $d_aF_a(r+a)$ includes the atom at $a=-r$ when $r<0$;
it cannot be replaced by only an ordinary derivative.
The corresponding positive-slope chain has transition distribution
\begin{equation}\label{eq:newconstantslopekernel}
 T(a,[0,b])=\frac1a\int_{-b}^aG_a(r)F_b(r+b)\,dr,
 \qquad a,b>0.
\end{equation}

Let $\lambda$ be the stationary slope law induced by $\nu_\infty$, so that
$\lambda([0,a])=\int F_a(r+a)\,\nu_\infty(dr)$.
Alternating the margin-to-slope and slope-to-margin transitions and using
$\nu_\infty K=\nu_\infty$ gives $\lambda T=\lambda$.
The conditional reward at slope $a$ is
\begin{equation}\label{eq:newconstantslopereward}
 h(a)=\frac4a\int_0^a V(r)G_a(r)\,dr,\qquad
 \gamma_{\rm Swing}=\int_0^\infty h(a)\,\lambda(da).
\end{equation}
This follows directly from~\eqref{eq:newconstantoutput} and Tonelli's theorem.
There is also a boundary representation:
\begin{equation}\label{eq:newconstantboundarydensity}
 \lim_{a\downarrow0}\frac{\lambda([0,a])}{a}
 =\frac{\gamma_{\rm Swing}}2.
\end{equation}
Indeed, the small-drift survival formula~\eqref{eq:smalldrift} gives
$F_a(r+a)/a\to2V(r)$ for fixed $r\ge0$.
For $r<0$ only $r\ge-a$ can contribute, and Lemma~\ref{lem:ballot} gives
the global bound $F_a(r+a)/a\le C(1+r_+)$ for $0<a\le1/2$.
Since $\nu_\infty$ has no atom at zero and has integrable positive part,
dominated convergence proves~\eqref{eq:newconstantboundarydensity}.
Here the boundary density means precisely the right derivative of the
distribution function, without an additional assumption of pointwise
smoothness of a stationary density.
The two representations provide distinct checks on deterministic integration;
the convergent construction below does not take $\lambda$ as an input.

\subsection{Positive Gaussian input formulas}
Spitzer's identity, in the full transform form proved in
Baxter~\cite[Example 2, eq.~(3.12)]{baxter1960}, gives
\begin{align}
 p_a&=\exp\!\left[-\sum_{n\ge1}\frac{\Phi(-a\sqrt n)}n\right],
 \label{eq:newconstantp}\\
 \E M_a&=\sum_{n\ge1}\left[
 \frac{\varphi(a\sqrt n)}{\sqrt n}-a\Phi(-a\sqrt n)\right].
 \label{eq:newconstantmaximummean}
\end{align}
Define the convolution exponential on $[0,\infty)$ by
\begin{equation}\label{eq:constantconvolution}
 \ell_a(x)=\sum_{n\ge1}\frac{\varphi((x+an)/\sqrt n)}{n^{3/2}},
 \qquad
 \mathcal E_a=\delta_0+\sum_{j\ge1}\frac{(\ell_a(x)dx)^{*j}}{j!}.
\end{equation}
Then $\operatorname{Law}(M_a)=p_a\mathcal E_a$ for $a>0$, and
\begin{equation}\label{eq:constantVinput}
 V(x)=\frac{\mathcal E_0([0,x])}{\sqrt2},\qquad x\ge0.
\end{equation}
To justify the zero-drift formula, $\ell_a\uparrow\ell_0$ locally, with
$\ell_0\le\zeta(3/2)/\sqrt{2\pi}$, so the convolution series converges
locally uniformly in cumulative mass. The small-drift survival
limit~\eqref{eq:smalldrift}, together with $p_a/a\to\sqrt2$, identifies
this limit with $\sqrt2V$. The shift in that survival limit can be removed
on compact positive intervals by its local uniformity; at zero the
convolution atom gives $V(0)=1/\sqrt2$.
Equivalently, Gaussian change of measure at the ladder epoch gives
$\operatorname{Law}(M_a)(dx)=p_ae^{-2ax}\sum_{j\ge0}
\operatorname{Law}(H_a)^{*j}(dx)$, whose zero-drift limit is the renewal
identity for $V$. Only local finiteness of $\mathcal E_0$ is required.
The small-drift formulas used below are
\begin{equation}\label{eq:newconstantsmalldrift}
 p_a\sim\sqrt2a,\qquad
 \E M_a=\frac1{2a}-\rho_{\rm S}+O(a),\qquad
 \rho_{\rm S}=-\frac{\zeta(1/2)}{\sqrt{2\pi}},
\end{equation}
as in Janssen and van Leeuwaarden~\cite{jvl}, Theorems 1--2.
All the inputs to~\eqref{eq:newconstantslopekernel} and
\eqref{eq:newconstantslopereward} are thereby specified.

\subsection{A contraction and one convergent evaluation}
Write $W_1$ for the first Wasserstein distance. Set
\[
 L_h=4+\frac8{\sqrt\pi},\qquad
 g_*=1+\frac{e^{-1}}{2\sqrt{2\pi}}
              \frac{e^{-1/2}}{1-e^{-1/2}},\qquad M_0=\rho_{\rm S}+g_*.
\]
These are explicit bounds, not additional coefficients to estimate.
The evaluation index $m$ below counts applications of the limiting
operator $T$. Its input zero is a chosen slope boundary, not the state
of the second segment in an actual run. The resulting convergence rate
concerns computation of the coefficient; it does not by itself bound
finite-$E$ convergence in the native segment index $k$.

\begin{theorem}\label{thm:constantcontraction}
The margin law $B_a(dr)=G_a(r)\ind_{\{r<a\}}dr/a$ increases in stochastic
order with $a$, while $T$ reverses stochastic order. Moreover,
\begin{equation}\label{eq:constantcontraction}
 W_1(T(a,\cdot),T(c,\cdot))\le\tfrac34|a-c|.
\end{equation}
Both $T$ and $h$ extend continuously to zero, with $h(0)=0$.
For $q_m=(T^mh)(0)$,
\begin{equation}\label{eq:constantevaluation}
 q_{2m}\uparrow\gamma_{\rm Swing},\qquad
 q_{2m+1}\downarrow\gamma_{\rm Swing},\qquad
 |q_m-\gamma_{\rm Swing}|\le L_hM_0(3/4)^m.
\end{equation}
The adjacent brackets have width at most $L_hM_0(3/4)^{2m}$.
\end{theorem}
The proof combines the stochastic increase of the intermediate margin
with a bound $1/2\le\frac d{da}\E_{B_a}R\le3/4$. Coupling ordered margins
and then applying the 1-Lipschitz slope map gives the contraction.
The Lipschitz reward converts it into a coefficient error; order reversal
gives the even and odd brackets. Details are in Appendix~\ref{app:evaluation}.

This is a Wasserstein bound for the limiting kernel, not a total variation
rate or a finite-$E$ mixing bound. The remainder controls exact operator
iteration. Numerical quadrature, inner series, domain truncation and
rounding errors must be controlled separately.

\subsection{Deterministic evaluation of the limiting operator}
\label{sec:constantnumerics}

We evaluate the slope kernel~\eqref{eq:newconstantslopekernel} directly,
without finite-$E$ simulations or Monte Carlo samples. We use the positive
convolution formula~\eqref{eq:constantconvolution}
and retain the atom explicitly. The implementation sums the first 64 terms
of $\ell_a$, evaluates the remaining sum by Euler--Maclaurin quadrature,
and forms its convolution exponential by a tilted, zero-padded FFT.
Gaussian convolution then gives $G_a$. Independent Poisson summation of
$\ell_a$ agrees at selected points within $1.2\times10^{-15}$.

We use two outer discretizations: midpoint cells with differences of the
transition CDF, and Chebyshev differentiation of that CDF with
Gauss--Legendre quadrature. A normalized left dominant eigenvector gives
the approximate stationary slope law. The conditional reward
$(4/a)\int_0^aV(r)G_a(r)\,dr$ is then integrated against that law.
Separately, twice the derivative at zero of its CDF gives the same
coefficient in the exact boundary identity. On the finite grid, the
reported estimate is $2\partial_b[m_hQ_h(\cdot,b)]_{b=0}$, where
$Q_h$ interpolates the transition CDF and $m_hP_h=\lambda_hm_h$,
$m_h\mathbf1=1$. This uses the unnormalized one-step output CDF;
division by $\lambda_h$ gives the boundary estimate corresponding to the
normalized spectral density. It does not call the solver for $V$.

For the reward calculation we use $V(x)=x+e(x)$, where first-step
decomposition of the zero-drift overshoot gives
\[
 e(x)=\varphi(x)-x\Phi(-x)+\int_0^\infty e(y)\varphi(y-x)\,dy.
\]
The implementation replaces $e(y)$ by $\rho_{\rm S}$ for $y>24$,
solves the resulting equation with 384 Gauss--Legendre nodes on $(0,24)$,
and tabulates $V$ on $[0,12]$ at spacing $0.001$ for cubic-spline
interpolation. The substituted tail is
$\rho_{\rm S}\Phi(x-24)$. This truncation is a numerical approximation,
separate from the exact convolution formula~\eqref{eq:constantVinput}.

\begin{table}[ht]
\centering\small
\begin{tabular}{rrr}
\toprule
Maximum-coordinate mesh & Reward integral & Output-CDF boundary estimate\\
\midrule
$0.0025$   & $1.812070582079$ & $1.812071601731$\\
$0.00125$  & $1.812070394806$ & $1.812070649915$\\
$0.000625$ & $1.812070348001$ & $1.812070411803$\\
\midrule
Last-pair extrapolation & $1.812070332400$ & $1.812070332432$\\
\bottomrule
\end{tabular}
\caption{Deterministic evaluations of the limiting kernel. The slope domain
is $(0,7)$, with 64 spectral nodes and 96 inner quadrature nodes. The last
row is second-order Richardson extrapolation, not a certified interval.}
\label{tab:constantdeterministic}
\end{table}

At maximum-coordinate mesh $0.005$ on slope domain $(0,6)$, changing the
spectral order from 32 to 64 to 96 changes the coefficient by less than
$3\times10^{-14}$. Expanding the slope domain from $(0,6)$ to $(0,7)$
changes it by approximately $2.1\times10^{-9}$; expansion from $(0,7)$ to
$(0,8)$ changes it by about $2.8\times10^{-12}$ at a shared finer mesh.
The midpoint-cell calculation converges independently toward the same value.

As a further check on the maximum distribution, a separate solver directly
discretizes the Lindley density equation
\[
 f_a(x)=p_a\varphi(x+a)+\int_0^\infty
             f_a(y)\varphi(x+a-y)\,dy,\qquad x>0.
\]
This check uses neither the compound-Poisson representation nor FFTs.
It shares the separately checked atom $p_a$ with the main solver.
At $a=0.5,1,2$, a 512-node Gauss--Legendre rule on $(0,24/a)$ agrees with
the extrapolated FFT values of $G_a$ within $4.3\times10^{-13}$ at the
tested points. These comparisons are numerical diagnostics, not rigorous
bounds for every drift or argument.

The finest matrix in Table~\ref{tab:constantdeterministic} has maximum
absolute row-mass defect approximately $4.79\times10^{-8}$ and dominant
eigenvalue $1.00000003041947$. Normalizing the eigenvector does not remove
the operator error; dividing the boundary estimate by this eigenvalue
changes it by about $5.51\times10^{-8}$. No negative matrix entries are
silently clipped.
The observed convergence supports the numerical value
\[
 \gamma_{\rm Swing}\approx1.8120703.
\]
It does not certify the displayed extrapolation digits. In particular,
the exact iteration remainder in
Theorem~\ref{thm:constantcontraction} does not include the numerical
integration, transform approximation, domain truncation, or rounding
errors of this implementation. The accompanying source bundle contains
the solver, saved matrices and a numerical README with commands for
reproducing the reward, boundary estimates, extrapolation and Lindley check.

\section{Transient means and the limiting coefficient}
\label{app:transientconvergence}

We prove the logarithmic mean and contraction statements of
Theorem~\ref{thm:transientconvergence} for the actual initialization.
Write $R_{j,E}=E-|X_j|$ for the margin after segment $j$ and
$N_j=L_{j,E}$ for its length. Segment $j$ starts from $R_{j-1,E}$;
its future increments are fresh after the true input rejection.

\subsection{Propagating the first-restart moment bounds}

Lemma~\ref{lem:firstmarginmoments} supplies the base case
$\sup_{E\ge E_0}\E|R_{1,E}|^p<\infty$ for every fixed $p>0$.
Suppose the same holds for the input to segment $j\ge2$.
The uniform conditional mean bound~\eqref{eq:meanuniform} gives
$\E N_j\le C_jE\log E$.
Let $\mathcal G_j$ be the predictable favorable-rejection branch defined
in Appendix~\ref{app:margins}. The pathwise implication
\eqref{eq:stationarypathgood}, Markov's inequality for the input margin,
and the stopped-square maximal inequality give
\begin{equation}\label{eq:transient-badbranch}
 \varepsilon_{j,E}:=\Prob(\mathcal G_j^c)
 \le\Prob(|R_{j-1,E}|>E/40)+C\E N_j/E^2
 \le C_j\frac{\log E}{E}.
\end{equation}
No stationary input law is needed here.

The tail calculation on that branch is also conditional on an accepted
prefix, hence applies to this initialization. For a lower rejection with
old slope $m\in[0,E]$ and last residual $-E+g$, $g\ge0$, rejection means
$Z<m-g$. The exterior output margin is $(-g-Z)_+$, so
\[
 \frac{\Prob(Z<-g-t)}{\Prob(Z<m-g)}
 \le\frac{\Phi(-g-t)}{\Phi(-g)}\le e^{-t^2/2}.
\]
Reflect for upper rejection and sum over the disjoint completion events.
On the complementary branch, the exterior margin is at most the absolute
selected increment. Predictable summation therefore yields
\begin{equation}\label{eq:transient-refinedtail}
 \Prob((R_{j,E})_->t)
 \le e^{-t^2/2}
 +\min\{\varepsilon_{j,E},\,2\E N_j e^{-t^2/2}\}.
\end{equation}
The sufficient path event in~\eqref{eq:stationarypathgood} is used only
to estimate the branch probability, not to condition a future increment.

Integrate~\eqref{eq:transient-refinedtail}, splitting at
$\sqrt{8\log E}$. For every fixed $p>0$,
\begin{equation}\label{eq:transient-negativemoments}
 \limsup_{E\to\infty}\E(R_{j,E})_-^p
 \le p\int_0^\infty t^{p-1}e^{-t^2/2}\,dt,\qquad j\ge2.
\end{equation}
Indeed, the extra integral below the cutoff is
$O_{j,p}((\log E)^{1+p/2}/E)$; the Gaussian tail above it also vanishes.
The pointwise bound~\eqref{eq:stationarymarginpositive} gives
\[
 (R_{j,E})_+\le5(R_{j-1,E})_-+10G_j,
 \qquad G_j=\sup_{n\ge0}\frac{|W_n|}{n+1},
\]
where $G_j$ has Gaussian tails. This completes induction and proves
\begin{equation}\label{eq:transient-fullmoments}
 \limsup_{E\to\infty}\E|R_{j,E}|^p<\infty,
 \qquad j\ge1,\quad p>0.
\end{equation}
Constants may depend on the fixed $j,p$. In particular,
$\E L_{k,E}=O_k(E\log E)$ for all fixed $k\ge2$; the second segment
has the stronger bound in Theorem~\ref{thm:secondsegment}.

\subsection{A positive logarithmic contribution from the third segment}

From any sequence $E\to\infty$, extract a subsequence such that
$R_{1,E}\Rightarrow\rho_1$. Only this weak convergence is needed for the mean-passage argument.
Compact-uniform kernel convergence~\eqref{eq:marginkernellocalconvergence},
tightness and the Feller property give, successively,
\begin{equation}\label{eq:transient-subsequentiallaws}
 R_{j,E}\Rightarrow\rho_1K^{j-1},\qquad j\ge1,
\end{equation}
for each fixed $j$ along that subsequence. For one step, test against a
bounded Lipschitz function, restrict the input to a compact set, and use
uniform convergence there; continuity of $Kf$ handles the limiting input.

Put $v(r)=4V(r)\ind_{\{r>0\}}$. A single application of $K$ has a
strictly positive density. Therefore $\rho_1K^{j-1}$ has no atom at
zero and has positive mass on every positive interval when $j\ge2$.
For each fixed $k\ge3$, the conditional mean theorem and
\eqref{eq:transient-fullmoments} give
\begin{equation}\label{eq:transient-subsequentialmeans}
 \frac{\E L_{k,E}}{E\log E}
 \longrightarrow\rho_1K^{k-2}v
 \quad\hbox{along the chosen subsequence}.
\end{equation}
To justify averaging, use compact-uniform conditional convergence on
each side of zero, the vanishing limiting mass of a shrinking zero
neighborhood, and the bound $g_E(r)/(E\log E)\le C(1+r_+)$.
Any moment of order greater than one in~\eqref{eq:transient-fullmoments}
controls the remote positive tail; tightness controls the negative tail.

The positive lower bound is uniform over possible subsequences for each
fixed $k\ge3$. Choose $M$ such that every possible $\rho_1$ puts mass
at least $1/2$ on $[-M,M]$. For a nonzero continuous function
$0\le\psi\le1$ supported in $(1,2)$, strict positivity and the Feller
property imply
\[
 \min_{|r|\le M}K^{k-2}\psi(r)>0.
\]
Since $v\ge c\psi$ for some $c>0$,~\eqref{eq:transient-subsequentialmeans}
has a positive lower bound depending only on $k$. This proves
$\E L_{k,E}=\Theta_k(E\log E)$ for every fixed $k\ge3$.

\subsection{A geometric bound on all subsequential coefficients}

Let $\mathsf A(r,da)=\operatorname{Law}(A_r)(da)$ and
$\mathsf B(a,dr)=\operatorname{Law}(a-O_a)(dr)$, so
$K=\mathsf A\mathsf B$, $T=\mathsf B\mathsf A$ and $h=\mathsf Bv$.
With $\bar{\alpha}=\rho_1\mathsf A$, the coefficient above is
\begin{equation}\label{eq:transient-slopecoefficient}
 \rho_1K^{k-2}v=\bar{\alpha}T^{k-3}h,\qquad k\ge3.
\end{equation}
Lemma~\ref{lem:firstmarginmoments} supplies a finite absolute constant
$C_1$ with $\int r_-\,\rho_1(dr)\le C_1$ for every possible $\rho_1$.
The bounds from Appendix~\ref{sec:newconstant},
$A_r\le r_-+G_+$, $\E G_+\le g_*$ and
$\int a\,\lambda(da)\le M_0$, give
\[
 W_1(\bar{\alpha},\lambda)\le C_1+g_*+M_0.
\]
Here $G_+=\sup_{n\ge0}W_n/(n+1)$ uses an independent Gaussian walk.
The contraction and the Lipschitz bound for $h$ therefore imply
\[
 |\bar{\alpha}T^{k-3}h-\lambda h|
 \le L_h(C_1+g_*+M_0)(3/4)^{k-3}.
\]
This bound is uniform over every subsequential first-restart law.
Consequently it proves the full-sequence statement
\begin{equation}\label{eq:transient-geometric-bound}
 \limsup_{E\to\infty}
 \left|\frac{\E L_{k,E}}{E\log E}-\gamma_{\rm Swing}\right|
 \le B(3/4)^{k-3},\qquad k\ge3,
\end{equation}
where $B=L_h(C_1+g_*+M_0)<\infty$ is independent of $k$. In particular,
\begin{equation}\label{eq:transient-iteratedlimit}
 \lim_{k\to\infty}\limsup_{E\to\infty}
 \left|\frac{\E L_{k,E}}{E\log E}-\gamma_{\rm Swing}\right|=0.
\end{equation}
The eventual tolerance thresholds may depend on the fixed index.
The argument does not require identifying the entrance law and does
not give a finite-$E$ mixing time. Each subsequential entrance
$\bar{\alpha}$ is distinct from the chosen zero-slope input used for
deterministic coefficient evaluation.

\section{Contraction of the coefficient-evaluation kernel}\label{app:evaluation}

\begin{proof}[Proof of Theorem~\ref{thm:constantcontraction}]
Since $\ell_a$ decreases with $a$, its positive convolution exponential
$\mathcal E_a$ decreases as a measure. Also
\[
 \frac d{da}\log\frac{p_a}{a}
 =\frac1{\sqrt{2\pi}}\sum_{n\ge1}n^{-1/2}e^{-a^2n/2}-\frac1a<0;
\]
the inequality compares a decreasing positive sum with its integral over
$(0,\infty)$. Thus $G_a(r)/a$ decreases at fixed $r$.
Integrating below $r<a$, and using that the CDF is one for $r\ge a$,
proves the stochastic increase of $B_a$.

The second-moment Lindley balance gives
\[
 \beta(a):=\E_{B_a}R=\frac a2-\frac1{2a}+\E M_a.
\]
For example, put $X=M_a+Z$ and integrate $rG_a(r)/a$ up to $a$;
$(X-a)_+\overset d=M_a$ cancels the second moments.
Differentiating~\eqref{eq:newconstantmaximummean} gives
\[
 \beta'(a)=\frac12+\frac1{2a^2}-\sum_{n\ge1}\Phi(-a\sqrt n).
\]
The function $x\mapsto\Phi(-a\sqrt x)$ is decreasing and convex, with
value $1/2$ at zero and integral $1/(2a^2)$. Rectangular and trapezoidal
bounds give $1/2\le\beta'(a)\le3/4$. An ordered quantile coupling has
$\E|R_c-R_a|=|\beta(c)-\beta(a)|\le(3/4)|c-a|$.
On a common independent Gaussian path,
$A_r=\sup_{n\ge0}(W_n-r)/(n+1)$ is decreasing and 1-Lipschitz in $r$.
This proves order reversal and~\eqref{eq:constantcontraction}.

The reward is increasing because $4V(r)\ind_{\{r>0\}}$ is increasing.
Differentiating its finite integral gives
\[
 0\le h'(a)\le4V(a)p_a/a\le L_h.
\]
Here $V(a)\le a+\sqrt{2/\pi}$ follows from the Gaussian overshoot bound;
$p_a\le1$ and $p_a/a\le\sqrt2$ follow from the preceding monotonicity
and~\eqref{eq:newconstantsmalldrift}. These bounds also imply $h(0)=0$.

Stochastic order makes $\E_{B_a}R_-$ decreasing. Since $\E R_+\le a/2$
and $\beta(a)\to-\rho_{\rm S}$, its limit at zero is $\rho_{\rm S}$; hence
$\E_{B_a}R_-\le\rho_{\rm S}$ for every $a$. Further $A_r\le r_-+G$, where
$G=\sup_{n\ge0}W_n/(n+1)$. Writing $c_n=(n+1)/\sqrt n$, a Gaussian
union bound gives
\[
 \E G\le1+\sum_{n\ge1}\int_1^\infty\Phi(-c_nt)dt
 \le1+\sum_{n\ge1}\frac{\varphi(c_n)}{c_n}\le g_*.
\]
Thus $T\operatorname{id}(a)\le M_0$, and $\lambda\operatorname{id}\le M_0$.
The Wasserstein estimate extends $T$ to zero by completeness of the
probability measures with finite first moment; this extension is
specified by the explicit kernel and introduces no unknown input law.
Iteration gives $W_1(\delta_0T^m,\lambda)\le(3/4)^mM_0$.
Pairing with the Lipschitz reward proves the error bound.
Finally $\delta_0$ is the smallest nonnegative input law, $T$ reverses
order and $T^2$ preserves it. Thus its even iterates increase and its odd
iterates decrease, with every even law below every odd law.
This proves the brackets. Their width follows from
$W_1(\delta_0T^{2m},\delta_0T^{2m+1})\le(3/4)^{2m}T\operatorname{id}(0)$.
\end{proof}

\end{document}